\documentclass[11pt]{article}
\usepackage{fullpage,graphicx,algorithmic,algorithm,enumerate,epsfig,wrapfig}
\usepackage{color,amssymb,amsthm,thmtools}
\usepackage{hyperref}
\usepackage{amsmath}
\usepackage{geometry}
\declaretheorem[parent=section]{theorem}

\declaretheorem[sibling=theorem]{definition}

\declaretheorem[sibling=theorem]{corollary}

\declaretheorem[sibling=theorem]{lemma}

\declaretheorem[sibling=theorem]{example}

\declaretheorem[sibling=theorem]{fact}

\DeclareMathOperator {\Z}{\mathbb{Z}}

\DeclareMathOperator {\E}{\mathbb{E}}
\let\epsilon=\varepsilon

\newcommand{\poly}{\mathrm{poly}}

\newcommand{\q}{\widehat{z}}
\renewcommand{\d}{z}

\newcommand{\EMD}{\mathrm{EMD}}
\newcommand{\EC}{C} 
\newcommand{\PR}{L} 
\newcommand{\survive}{K} 
\newcommand{\R}{\mathbb{R}}
\newcommand{\pois}{\mathrm{Poisson}}

\title{Robust Set Reconciliation via Locality Sensitive Hashing}

\author{Michael Mitzenmacher\thanks{Harvard University School of Engineering and Applied Sciences.  email: michaelm@eecs.harvard.edu.  Michael Mitzenmacher was supported in part by NSF grants CNS-1228598, CCF-1320231, CCF-1563710 and CCF-1535795.} \and Tom Morgan\thanks{Harvard University School of Engineering and Applied Sciences.  email: tdmorgan@seas.harvard.edu.  Tom Morgan was supported in part by NSF grants CNS-1228598 and CCF-1320231.}}

\date{}

\begin{document}
\maketitle

\begin{abstract}
We consider variations of set reconciliation problems where two parties, Alice and Bob, each hold a set of points in a metric space, and the goal is for Bob to conclude with a set of points that is close to Alice's set of points in a well-defined way.  This setting has been referred to as robust set reconciliation.  More specifically, in one variation we examine the goal is for Bob to end with a set of points that is close to Alice's in earth mover's distance, and in another the goal is for Bob to have a point that is close to each of Alice's.  The first problem has been studied before;  our results scale better with the dimension of the space.  The second problem appears new.  

Our primary novelty is utilizing Invertible Bloom Lookup Tables in combination with locality sensitive hashing.  This combination allows us to cope with the geometric setting in a communication-efficient manner.  
\end{abstract}

\section{Introduction}

\paragraph{Set reconciliation.} Set reconciliation is a basic paradigm for data synchronization, for distributed databases and other distributed systems.  In a standard set reconciliation problem, two parties, Alice and Bob, hold sets of elements from a shared universe, and the goal is for them to communicate so that one or both of them has the union of both sets.  In many settings, the size of the set difference may be very small even though the sets may be large, and therefore the goal is for the communication to be proportional to the size of the set difference, rather than the size of each of the sets (which could be achieved simply by transferring the sets themselves).  We provide further description and references below.    

\paragraph{Robust set reconciliation.} Robust set reconciliation, introduced in \cite{chen2014robust}, generalizes set reconciliation to the scenario where the set elements lie in a metric space, and sufficiently close points should be thought of as equal.  As a natural example, set elements might be geometric coordinates for objects, as determined by sensors.  Each sensor corresponds to a set, and for the same object, each sensor might have slightly different, noisy measurements.  We might wish two sensors to synchronize their collections of known objects, and objects within a certain measured distance are either assumed to be (or for practical purposes may be treated as) the same.  Other applications would include reconciling other potentially noisy data, such as databases with floating point measurements or calculations, or databases with image data that has been subjected to varying compression schemes.  In such cases, the databases would not end up with the same data;  but this would suffice for numerical data sets where having points that are close enough may be all that is needed \cite{chen2014robust}.  This would, for example, be useful when the databases are used for machine learning via clustering or nearest neighbor search.  Here the most valuable new data to reconcile would be the outliers.

We study two different models of robust set reconciliation that achieve different types of guarantees.  In both settings Alice has a point set $S_A$, and Bob has a point set $S_B$ where all the points lie in a metric space $(U, f)$, where $U$ is a discretized metric space (as our bounds may depend on $|U|$) such as $[\Delta]^d$.  They communicate over a constant number of rounds so as to \emph{reconcile} Alice's data with Bob's;  that is, Bob's final point set $S_B'$ is close to Alice's, where the notion of closeness depends on the model.  We limit ourselves to computationally efficient (polynomial time) protocols, however what can be achieved without this limitation is an interesting open question.

\paragraph{Earth Mover's Distance model.} The first model, the Earth Mover's Distance model, was originally introduced in \cite{chen2014robust}.  As in \cite{chen2014robust}, we restrict ourselves to metric spaces of the form $(U,f)=([\Delta]^d,\ell_p)$.  We require that $|S_A| = |S_B| = |S_B'|$.  The goal here is for Bob to compute an $S_B'$ which minimizes $\EMD(S_A, S_B')$, the earth mover's distance between $S_A$ and $S_B'$, with only reasonable amounts of communication.\footnote{We note that Definition 2 of \cite{chen2014robust} makes the additional stipulation that $S_B' \subset S_A \cup S_B$, however neither our protocol nor the protocol of \cite{chen2014robust} meet this requirement. Both include points in $S_B'$ that approximate, without necessarily equaling, points from $S_A$.}  The earth mover's distance is the min-cost perfect matching between the point sets, where the cost is the distance function $f$.  

The following notation will be helpful.  Let $\EMD_k(X, Y)$ be the minimum earth mover's distance achievable between $X$ and $Y$ after excluding $k$ points from each set.  In other words, we would obtain $\EMD(S_A, S_B') = \EMD_k(S_A,S_B)$ if we were able to exactly identify the optimal $k$ points to remove from $S_B$ and the optimal $k$ points from $S_A$ to replace those with. 
Given a communication bound of $O(k \log |U|)$ bits (where $k$ is an input parameter), the smallest $\EMD(S_A, S_B')$ one could reasonably hope to achieve is $\EMD(S_A, S_B') = \EMD_k(S_A,S_B)$.   Indeed, \cite{chen2014robust} provided lower bounds for this model, which confirm that achieving $\EMD(S_A, S_B') = \EMD_k(S_A,S_B)$ requires $\tilde{\Omega}(k \log |U|)$ bits of communication. 

We do not achieve $\EMD(S_A, S_B') = \EMD_k(S_A,S_B)$, but instead obtain a multiplicative approximation to it while using $\tilde{O}(k)$ communication.\footnote{The $\tilde{O}$ here hides log factors of $n$ and log factors of parameters depending on the metric space, in particular $|U|$.}  In particular, we achieve an $O(\log n)$ approximation, improving over the $O(d)$ approximation (where $d$ is the dimension) of \cite{chen2014robust} for high dimensional data.  (One might think the results of \cite{chen2014robust}, in combination with dimension-reduction via the Johnson-Lindenstrauss lemma \cite{johnson1984extensions}, would achieve this for, for example, the $\ell_2$ norm.  However, inverting the dimensionality reduction would require additional rounds of communication;  moreover, our result holds for metrics where such general dimensionality reduction does not exist, such as the $\ell_1$ norm \cite{andoni2011near}.)  The setting where our improvement is most obvious is for Hamming space, where an $O(d)$ approximation is essentially useless, as the space has diameter $d$, while an $O(\log n)$ approximation would be useful, for example, when $n = \poly(d)$.  
Our results for this model are presented in \autoref{sec:emd}.

\paragraph{Gap Guarantee model.} In the second model, which we introduce, we aim for a stronger guarantee of closeness for every point, and consider the necessary communication.  Here Bob's final point set $S_B'$ will be of the form $S_B \cup T_A$, where $T_A \subset S_A$ includes every point in $S_A$ which is at least some chosen distance $r_2$ from every point in $S_B$.  Note that $T_A$ is allowed to contain additional points from $S_A$ beyond these.  That is, Bob is guaranteed that every point in the union of Alice's and Bob's original sets is close to some point in his final set.  In order to achieve nontrivial communication bounds for this guarantee, we introduce an additional parameter $r_1 < r_2$, with the intuition being that most of the points in $S_A$ are already within a distance $r_1$ of some points in $S_B$.  Our communication bounds are then in terms of the number of points that are not within $r_1$, and the gap between $r_1$ and $r_2$.  We call this model the Gap Guarantee model, and study it in \autoref{sec:gap}.  

We believe this model to be quite natural given our motivating sensor network example.  We would expect sensors observing the same object the have similar measurements (below some distance $r_1$) while discrete objects would yield very different measurements (above some distance $r_2$).  This model then guarantees the recovery of all differing objects, while the Earth Mover's Distance model gives a much weaker guarantee on the recovered set.  However, to achieve this distinct guarantee we may require significantly more communication.

Our general protocol uses $O((k + \rho n) \poly \log n + k \log |U|)$ bits of communication, where $k$ is a bound on the number of points each party has that are more than $r_1$ from any of the other party's points and $\rho$ is a parameter of the locality sensitive hash family used in the protocol, which depends on $r_1$ and $r_2$.  (In many metric spaces, $\rho = r_1/r_2$;  we will explain further in context.)  The improvement this achieves over the naive $O(n \log |U|)$ communication is twofold: its dependence on $\log |U|$ is proportional to $k$ and not $n$, which is very relevant for high dimensional data (where $\log |U|$ may be linear in the dimension $d$), and for a sufficiently small (sub-constant) $\rho$, it yields sublinear total communication.

\paragraph{One-way reconciliation.} Both of our models are defined for \emph{one-way} reconciliation, which we define to mean that which Bob wants to conclude with something approximating Alice's data, but Alice makes no changes to her own data.  For standard set reconciliation, the two-way reconciliation problem is natural, as we can have have both parties conclude with the union of their original sets.  For robust set reconciliation problems, the one-way variation is more natural.  For example, for both models we consider, we can easily achieve a natural version of two-way reconciliation by having both Alice and Bob run the protocol once in each direction;  however, they will generally not end with the same point set.  Furthermore, it is unclear what the natural guarantee for a two-way version of the Earth Mover's Distance model would be, especially since we don't expect Alice and Bob to end with the same set. 

\subsection{Related Work}
Here, we briefly describe important related work.  Standard set reconciliation has been studied in the context of distributed synchronization, with many possible applications, see e.g. \cite{minsky2003set,  starobinski2003efficient,eppstein2011s,mitzenmacher2013simple,mitzenmacher2012complexity} and citations therein.  As a fairly recent example, IBLTs (described below) have been proffered as a technique for scalable synchronization of transactions for Bitcoin, and have been discussed as an addition to the Bitcoin protocol \cite{gistgithub}.  

Two technologies underlying our results are locality sensitive hashing and invertible Bloom lookup tables (IBLTs).  Locality sensitive hashing hashes items that are close into the same bucket.  Here we follow the framework established by Indyk and Motwani \cite{indyk1998approximate}, though countless other works in locality sensitive hashing provide inspiration \cite{charikar2002similarity,datar2004locality,andoni2006near}.  Invertible Bloom lookup tables \cite{eppstein2011s,eppstein2011straggler,gm11} provide a particularly efficient approach for standard set reconciliation;  they allow sets with $d$ differences to be synchronized, after preprocessing taking time linear in the set sizes, in $O(d)$ space and time with some small probability of error.  We describe IBLTs in more detail below.  One of our primary technical contributions in this paper is an analysis of how errors due to noisy or otherwise inexact data propagate when using IBLTs, as we wish to limit this effect;  our analysis here may prove useful for other work.  

The idea of using hash-based data structures to handle close matches appears in the work of Kirsch and Mitzenmacher \cite{kirsch2006distance}, who consider generalizing Bloom filters (for membership queries) to distance-sensitive Bloom filters by making use of locality-sensitive hash functions to return a positive result if a query is close to a set element. Chen et. al. \cite{chen2014robust} introduce the concept of robust set reconciliation, and use a randomly offset quadtree with IBLTs to develop protocols for the earth mover's distance problem variation we consider here.  
Chen et. al. discuss many reasons why numerical data sets may have slightly different values, including noise, lossy compression, rounding errors, and privacy-preserving transformations.  Applications for the settings they describe are similarly relevant for our results.   

A related problem to our Earth Mover's Distance model is that of sketching and estimating the earth mover's distance \cite{charikar2002similarity,andoni2008earth,andoni2009efficient}.  However, we did not find existing results and techniques on this problem to be useful for robust set reconciliation, nor did the techniques we developed yield any immediate results in the sketching/estimation regime.

We also make use of the recent work of Mitzenmacher and Morgan \cite{mitzenmacher2017reconciling} on reconciling sets of sets.  In this setting, Alice and Bob each hold a parent set filled with child sets, and the goal is to synchronize their sets of sets using communication proportional to the number of child set operations by which they differ.  This model generalizes to reconciling various other sets of objects such as lists and unlabeled graphs.

Some of our analysis requires some technology from the theory of branching processes;  here \cite{geiger1999elementary,jiang2014parallel} proved helpful.

As mentioned, \cite{chen2014robust} is the most closely related work.  Indeed, like \cite{chen2014robust} we utilize locality-sensitive hashing in combination with IBLTs.  We differ in that, as mentioned, Chen et. al. specifically use a randomly offset quadtree, while we allow for any of a large class of locality sensitive hash families.  We call this class \emph{multi-scale} locality sensitive hash families, and they have the property that the probability of collision between two points gracefully degrades as a function of the points' distance.  Our main innovation comes from how we use our IBLTs.  \cite{chen2014robust} simply rounds points to the center of their quadtree cell, and insert those into an IBLT, while we insert key-value pairs where the key is a point's locality sensitive hash value and the value is the point itself.  Handling these pairs, which may have different values for the same key, requires a robust variant of an IBLT, along with some in depth analysis of an IBLT's peeling process.  We expect these ideas (multi-scale locality sensitive hashing and robust IBLTs) to be more generally useful.

\section{Preliminaries / Techniques Used}
We assume throughout that Alice and Bob's data points lie in a metric space $(U, f)$.  For technical simplicity, we often assume that $U = [\Delta]^d$ for some $\Delta, d \in \Z^+$, and that Alice and Bob have equal-sized point sets.  Specifically, Alice and Bob have point sets $S_A, S_B \subset U$ respectively, and $|S_A| = |S_B| = n$.

We work in the word RAM model, with words of size $\Omega(\log n + \log \Delta)$. All protocols are performed assuming public coins, meaning that the random bits used are shared by Alice and Bob without requiring any communication.  This in particular allows us to assume that all hash functions are shared between Alice and Bob, without worrying about the cost required to communicate them.  There are standard methods for converting protocols with public coins into ones with private coins using minimal additional communication \cite{newman1991private}.  In practice, one can often approximate protocols with public coins by first sharing a small random seed.  We sometimes refer to the number of rounds of communication a protocol uses, which is equal to the number of messages sent.  In particular, a protocol using only one round consists of a single message from Alice to Bob (or vice-versa).

\subsection{Locality Sensitive Hash Functions}
We start with the standard definition of locality sensitive hashing \cite{indyk1998approximate}.

\begin{definition}[LSH] A family $\mathcal{H} \subseteq \{h \mid h\colon U \to V\}$ is a \emph{locality sensitive hash} (LSH) family with respect to $(U,f)$ with parameters $(r_1, r_2, p_1, p_2)$ if $r_1 < r_2$, $p_1 > p_2$ and for any $x, y \in U$,
\begin{itemize}
\item if $f(x, y) \leq r_1$ then $\Pr_{h \sim H} [h(x) = h(y)] \geq p_1$, and
\item if $f(x, y) > r_2$ then $\Pr_{h \sim H} [h(x) = h(y)] \leq p_2$.
\end{itemize}
\end{definition}

A commonly defined meta-parameter for locality sensitive hash functions is $\rho = \log p_1 / \log p_2$, which is the key parameter of in interest in the analysis of many approximate nearest neighbor algorithms, and appears in our analysis as well.  It is known for example that there exist LSH families for the $\ell_1$ metric with $\rho = \Theta(r_1/r_2)$ (\cite{datar2004locality}) and for $\ell_2$ with $\rho = \Theta((r_1/r_2)^2)$ (\cite{andoni2006near}).

For some of our results, we require a slightly stronger formulation.  We have not found this formulation in the literature, although related ideas can be found in for example \cite{charikar2002similarity}, which includes a condition that has the probability that two hash values collide fall with their distance in a natural way.

\begin{definition}[MLSH] A family $\mathcal{H} \subseteq \{h \mid h\colon U \to V\}$ is a \emph{multi-scale locality sensitive hash} (MLSH) family with respect to $(U,f)$ with parameters $(r, p, \alpha)$ if $r > 0$, $0 < p < 1$, $0 < \alpha < 1$, and, for any $x, y \in U$,
\begin{itemize}
\item $\Pr_{h \sim \mathcal{H}} [h(x) = h(y)] \leq p^{\alpha \cdot f(x, y)}$, and
\item if $f(x, y) \leq r$ then $p^{f(x, y)} \leq \Pr_{h \sim \mathcal{H}} [h(x) = h(y)].$
\end{itemize}
\end{definition}

Many standard LSH families are also MLSH families for the right setting of their parameters.  One simple example is the standard LSH family for $(\{0,1\}^d, f_H)$ where $f_H$ is the Hamming distance.  The standard LSH here simply samples a random bit from the input.  The probability of collision between $x, y \in \{0,1\}^d$ is $1 - f_H(x,y) / d$ which is at most $e^{-f_H(x,y) / d}$ and at least $e^{-2f_H(x,y) / d}$ for $f_H(x,y) \leq .79 d$.  We can replace the $d$ in these bounds with any $w \geq d$ by padding our points with 0s until they are $w$-dimensional before sampling choosing a bit to sample.\footnote{Equivalently, and more efficiently, with probability $d/w$ our hash function will sample a random bit, and with probability $1-d/w$ it will be a constant function always equaling $0$.}  This yields the following lemma.
\begin{lemma} \label{lem:hammlsh}
For any $w \geq d$, there exists an MLSH family with respect to $(\{0,1\}^d, f_H)$ with parameters $(.79 w, e^{-2 / w}, 1/2)$.
\end{lemma}

Many other simple MLSH families exist.  For example,
inspection of simple random grid hashing and $p$-stable distribution hashing \cite{datar2004locality} yields the following lemmas, whose proofs are in \autoref{app:mlsh}.
\begin{restatable}{lemma}{lemgridmlsh}
\label{lem:gridmlsh}
For any $w > 0$, there exists an MLSH family with respect to $([\Delta]^d, \ell_1)$ with parameters $(.79 w, e^{-2 / w}, 1/2)$.
\end{restatable}
\begin{restatable}{lemma}{lempstablemlsh}
\label{lem:pstablemlsh}
For any $w > 0$, there exists an MLSH family with respect to $([\Delta]^d, \ell_2)$ with parameters $(.99w, e^{-2\sqrt{2/\pi}/w}, 1/(4\sqrt{2}))$.
\end{restatable}

\subsection{Invertible Bloom Lookup Tables}
\label{sec:iblt}

We briefly review the basic properties of IBLTs;  more details can be found in \cite{eppstein2011straggler,gm11}.  An IBLT is a hash table using $q$ hash functions and $m$ cells to store key-value pairs, where the keys and values are assumed to have a fixed-size representation.  (In cases where there are no associated values, IBLTs can be used to just hold keys.)  A key has $q$ associated hash values, with each hash value indexing a cell of the table.  (We assume these cells are distinct;  for example, one can use a partitioned hash table, with each hash function mapping to $m/q$ cells.)  Each cell maintains a count of the number of key-value pairs hashed to it, an XOR of all of the keys hashed to it, an XOR of the values hashed to it, and an XOR of checksums (e.g. fingerprints), one for each key hashed to it;  adding a key-value pair simply updates the values in the $q$ associated cells.  The checksum, obtained using another hash function, is sufficiently large so as to ensure that with high probability, none of the distinct keys' checksums collide.  Deleting a key from an IBLT is similar to adding it, except that now we decrement the counts instead of incrementing them.

We can find all the elements in an IBLT, or \emph{invert} it, using a peeling process, if $m$ is large enough compared to the number of key-value pairs stored.  Whenever a cell in the table has a count of 1, the XOR of the key values in that cell equals the key hashed to that cell, and similarly for the value, so we can recover and then delete them from the table.  Such deletions may yield more cells with a count of 1, allowing the process to continue until no keys remain in the table.  By viewing the IBLT as a random hypergraph with the cells being $m$ vertices and the keys corresponding to hyperedges of cardinality $q$, we can analyze this peeling process;  all key-value pairs are recovered unless the hypergraph has a nonempty 2-core, the probability of which can be directly bounded.  This gives the following theorem.
\begin{theorem}[Theorem 1 of \cite{gm11}] \label{thm:iblt}
There exists a constant $0 < c < 1$ so that an IBLT with $m$ cells and at most $cm$ keys will successfully extract all key-value pairs with probability at least $1 - O(1/\poly(m))$, and the process takes $O(m)$ time.
\end{theorem}

We can apply IBLTs to standard set reconciliation when Alice and Bob have an upper bound $d$ on the size of their set difference.  Bob constructs an $O(d)$ cell IBLT by adding each of his set elements to it.  (The elements can be treated as keys, no values are needed.)  He then sends it to Alice who deletes each of her set elements from it.  Note that after this process the only elements in the IBLT are from the set difference, as elements in both sets are added but then deleted.  Now cells with a count of 1 or $-1$ may hold a single element, but such a cell may also hold multiple elements;  a cell with a count of 1 may hold two elements from Bob and one from Alice.  The checksum can be used to double-check that a cell with a count of 1 corresponds to a single key.  Alice can therefore extract all the keys from the IBLT using a variation of the peeling process described above, and can reconcile the sets after finding this difference.  We will sometimes refer to this process of recovering a set difference from an IBLT as ``decoding'' it.

For some of our results we require a variation of the IBLT that we call a Robust Invertible Bloom Lookup Table (RIBLT).  The RIBLT differs from an IBLT in several ways (throughout, $n$ will refer to an upperbound on the number of key-value pairs inserted or deleted from the table):
\begin{enumerate}
\item The peeling occurs in a specific ``breadth-first,'' first-come first-served order.  By this we mean that if a cell (vertex) $u$ has a single key (hyperedge) remaining earlier in the process than another cell $v$, then $u$ must be peeled before $v$.
\item RIBLTs are sparser than IBLTs, but still require only a number of cells linear in the number of keys.  In particular, we require that $c < 1 / (q(q-1))$, which results in an underlying hypergraph that is all trees and unicyclic components with high probability. See \cite{eppstein2014wear} for a discussion of and definitions for trees an unicyclic components in the hypergraph.
\item Rather than each cell maintaining an XOR of all its keys and key checksums, it maintains a sum for each.  When we add an key-value pair to the table, we add the binary representation of the key to the key sum in each cell, and similarly for the checksum of the key.  This may require more space per cell to avoid overflow.  If the original universe of keys is $U$, and thus their binary representations take $O(\log |U|)$ bits, then we now need $O(\log(|U|n))$ bits to store each cell's key sum.
\item Similarly, rather than each cell maintaining an XOR of all its values, it maintains a sum.  We restrict ourselves to values from a universe of the form $[\Delta]^d$.  Now the sum of values stored in each cell will actually store a binary representation of a point from $\{-n\Delta,\ldots,n\Delta\}^d$.  To update a cell's value sum, we map the binary representation back to a point from this space, add or subtract the new value from that point, then re-encode the resulting point into binary.  This requires $O(d \log (n \Delta))$ bit per cell.
\item 
\newcommand{\cellcount}{\mathcal{C}}
\newcommand{\cellkey}{\mathcal{K}}
\newcommand{\cellvalue}{\mathcal{V}}
\newcommand{\cellchecksum}{\mathcal{S}}
\newcommand{\checksum}{\mathrm{checksum}}
These changes allows us to decode an RIBLT even when there are duplicate keys in the table.  Now rather than peeling a key from a cell only when a single key is mapped to that cell (and thus the count is $\pm 1$), we also peel when the multiset of keys mapped to a cell are all equal.  Let $\cellcount$ be the cell's count field, $\cellkey$ be its key sum field, $\cellvalue$ be its value sum field, $\cellchecksum$ its sum of checksums field, and $\checksum()$ be our checksum function.  We recognize that cell's contents correspond to copies of the same key when $\cellkey$ and $\cellvalue$ are divisible by $\cellcount$, and $\checksum(\cellkey/\cellcount)=\cellchecksum/\cellcount$.  If this occurs, then with high probability there are $\cellcount$ copies of the same key $\cellkey/\cellcount$ added to that cell.

To peel such a cell, as before we subtract (or add) its cell contents from each cell $\cellkey/\cellcount$ hashes to.  We then extracted $\cellcount$ key-value pairs, where each pair's key is $\cellkey/\cellcount$ and each pair's value is independently determined by the following procedure.  First we take $\cellvalue$ and interpret it as a point in $\R^d$.  We then divide each entry by $\cellcount$ and shift the result into $[0, \Delta]$ by changing entries less than $0$ to 0 and entries greater than $\Delta$ to $\Delta$.  We then take each entry not falling in $[\Delta]$ (those that aren't integers) and randomly round them up or down to the nearest integer, with probability of rounding equal to the fractional remainder. This guarantees that even when multiple pairs are added with the same key but differing values, the extracted pairs all have values from our desired $[\Delta]^d$ space.
\end{enumerate}

These modifications allow us to perform a more detailed analysis of the table's underlying hypergraph than \autoref{thm:iblt} above provides, as when we have noisy values our inversion process may accumulate the discrepancy between values that are ``close'' but not equal as we peel elements away in the IBLT.  We discuss this fully in \autoref{sec:emd}, where we utilize results from \cite{jiang2014parallel} that were used to analyze a parallel version of this peeling process.

\section{Earth Mover's Distance}
\label{sec:emd}

First we formally define the Earth Mover's Distance model.  
\begin{definition}[EMD model] Suppose Alice and Bob have sets of points, $S_A$ and $S_B$ respectively, from a metric space of the form $([\Delta]^d,\ell_q)$, and $|S_A| = |S_B| = n$.  The goal of the \emph{Earth Mover's Distance model} of robust set reconciliation is for Bob to find a point set $S_B' \subset U$, $|S_B'| = n$, such that the earth mover's distance $\EMD(S_A,S_B')$ is minimized while adhering to a given upper bound on communication.
\end{definition}

\begin{definition}[$\EMD$]
Given point sets $X = \{x_1,\ldots,x_n\}$ and $Y = \{y_1,\ldots,y_n\}$ from a metric space $(U,f)$,
\begin{equation*}\EMD(X, Y) = \min_{\mathrm{bijection} \; \pi : [n] \to [n]} \sum_{i=1}^n f(x_i, y_{\pi(i)}).\end{equation*}
\end{definition}

Our protocol will ultimately relate $\EMD(S_A,S_B')$ to $\EMD_k(S_A,S_B)$, which is the minimum achievable earth mover's distance between $S_A$ and $S_B$ after excluding $k$ points from each.
\begin{definition}[$\EMD_k$]
Given point sets $X = \{x_1,\ldots,x_n\}$ and $Y = \{y_1,\ldots,y_n\}$ from a metric space $(U,f)$,
\begin{equation*}\EMD_k(X, Y) =\min_{T \subset [n], |T|=n-k} \left(\min_{\mathrm{injection} \; \pi : T \to [n]} \sum_{i\in T} f(x_i, y_{\pi(i)})\right).\end{equation*}
\begin{equation*}\EMD_k(X, Y) =\min_{X'\subseteq X, Y'\subseteq Y,|X'| = |Y'|= n-k}\EMD(X',Y').\end{equation*}
\end{definition}

The basic idea behind our protocol for the Earth Mover's Distance model is that we use an MLSH family $\mathcal{H}$ to hash Alice and Bob's points at various different resolutions.  We achieve finer resolutions by concatenating more and more hash functions from $\mathcal{H}$, thus partition the $[\Delta]^d$ into progressively smaller regions.  For each of these resolutions, Alice sends Bob an RIBLT consisting of (key, value) pairs where the key is the hash of one of her points and the value is the point itself.  Bob deletes his hashed points from the RIBLTs, and then finds the highest resolution RIBLT which is decodable, and uses the decoded points extracted from it to form his $S_B'$.  Since non-equal points can have the same key (their hash value), they won't fully ``cancel'' when decoding the RIBLT, and thus the decoded points will have some error.  Much of our technical work is bounding this error.

In what follows we assume we have parameters $D_1$ and $D_2$ such that $D_1 \leq \EMD_k(S_A, S_B) \leq D_2$ and $\max_{a \in S_A, b \in S_B} f(a, b) < M.$  In the case where $q = 1$ and we have no prior knowledge about $S_A$ and $S_B$, we can simply use $D_1 = 1, D_2 = n \cdot d \cdot \Delta$ and $M = d \cdot \Delta$.  (Note that if $\EMD_k(S_A,S_B) = 0$, this problem can be solved exactly with a standard set reconciliation protocol, so it sensible to assume that $D_1 \geq 1$.)  We also require an MLSH family for $([\Delta]^d,\ell_q)$ with parameters $(r, p, \alpha)$ such that $r \geq \min(M, D_2)$ and $p \geq e^{-k/(24D_2)}$.\footnote{Note that given $p$ and $D_2$ we must choose an MLSH family with $p$ large enough to meet this condition.  All of our example MLSH families allow for arbitrarily large $p$ values, and in general it is easy to increase $p$ by adding some set of constant functions to the MLSH family.  It may be unintuitive that we would want to increase $p$ in this way, but it is necessary to avoid over-partitioning the space while still allowing us to use enough independent functions from our MLSH family that the probability of different pairs of points colliding is sufficiently independent for our analysis.}  The full protocol appears in Algorithm \ref{alg:emd}.

\begin{algorithm}[h]
\caption{EMD Protocol}
\label{alg:emd}
\begin{itemize}
\item Alice creates $t = \log_2\left(D_2 / D_1 \right)+1$ RIBLTs $T_1,\ldots,T_t$, each with $q \geq 3$ hash functions and $m = 4 q^2 k$ cells.
\item Alice draws $s=\frac{k}{8 D_1 \ln(1/p)}$ hash functions $g_1,\ldots,g_s$ from $\mathcal{H}$. She draws $h$ from a $2$-wise independent class of hash functions with range $\{0,1\}^{\Theta(\log n)}$.
\item For each $i \in \{1,\ldots,t\}$ and $a \in S_A$, Alice forms a key-value pair and inserts it into $T_i$.  The key is $\text{key}_i(a)=h\left(g_1(x),\ldots,g_{2^{i-1} s D_1 / D_2}(a)\right)$, and the value is $a$.
\item Alice sends $T_1,\ldots,T_t$ to Bob.
\item For each $i \in \{1,\ldots,t\}$ and $b \in S_B$, Bob deletes the pair $(\text{key}_i(b),b)$ from $T_i$.  (Note that he knows $g_1,\ldots,g_s$ and $h$ due to public coins.)
\item Bob finds $i^*$, the largest $i$ such that $T_i$ successfully decodes to at most $4 k$ key-value pairs ($2 k$ pairs per party).  Let $X_B$ be the values that $T_{i*}$ decodes from his side, and $X_A$ the values it decodes from Alice's side.  If no $T_i$ successfully decodes Bob reports failure.
\item Bob finds $Y_B$, the subset of $S_B$ matched in the min cost matching between $X_B$ and $S_B$.  He then outputs $S_B' = (S_B \setminus Y_B) \cup X_A$.
\end{itemize}
\end{algorithm}

\begin{theorem} \label{thm:emd}
Algorithm \ref{alg:emd} uses
$O\left(k d\log \left(\Delta n\right) \log\left(D_2 / D_1 \right)\right)$
bits of communication and
$$O\left(t n k / (D_1 \log(1/p)) + dn \log(D_2/D_1) +d n k + n k^2\right)$$
time.  If $\EMD_k(S_A,S_B) \leq D_2$, it reports failure with probability at most $1/8$.  If $\EMD_k(S_A,S_B) \geq D_1$ and it does not report failure, then 
$\EMD(S_A,S_B') \leq O(\alpha^{-1} \log n) \cdot \EMD_k(S_A, S_B)$ with probability at least $3/4$.
Here $t$ is an upper bound on the time to evaluate functions from $\mathcal{H}$.
\end{theorem}

Before proving this theorem, let us discuss its implications for some settings.  Suppose our metric space is $(\{0,1\}^d,f_H)$, and we have no assumptions on $D_1$, $D_2$, and $M$.  Applying the MLSH family of \autoref{lem:hammlsh} to \autoref{thm:emd} yields the following.\footnote{In order to meet the restriction that $p \geq e^{-k/(24D_2)}$ we choose $w = 48nd/k$ when applying \autoref{lem:hammlsh}.  This is already factored into the stated bounds of the corollary.}
\begin{corollary} \label{cor:ham}
There is a protocol for the Earth Mover's Distance model on $(\{0,1\}^d,f_H)$ using
$O\left(k d \log n \log(dn)\right)$
bits of communication, $O\left(dn^2+nk^2\right)$ time, and successfully computes $S_B'$ such that $\EMD(S_A,S_B') \leq O(\log n) \cdot \EMD_k(S_B, S_A).$  with probability at least $5/8$.
\end{corollary}

Now suppose we are working in $([\Delta]^d,\ell_2)$.  In such a case we can divide the range $[D_1,D_2]$ into $I=O(\log(D_2/D_1))$ intervals $[D_1^{(1)}, D_2^{(1)}], [D_1^{(2)}, D_2^{(2)}], \ldots, [D_1^{(I)}, D_2^{(I)}]$ such that $D_1^{(1)} = D_1$, $D_2^{(I)} = D_2$, and for all $j$, $D_2^{(j)}/D_1^{(j)} = O(1)$ and $D_1^{(j+1)}=D_2^{(j)}$.  We run Algorithm \ref{alg:emd} in parallel for each of these intervals, and have Bob use the output of version for the smallest index interval which did not report failure.  For the $j$th interval, we use the MLSH family of \autoref{lem:pstablemlsh} (with $w = \Theta(\min(M,D^{(j)}_2) + D_2^{(j)}/k)$) yields the following bounds.
\begin{corollary}
There is a protocol for the Earth Mover's Distance model on $([\Delta]^d,\ell_2)$ using
$O\left(k d\log(n\Delta) \log(D_2/D_1)\right)$
bits of communication, $O\left((dnk + nk^2)\log(D_2/D_1)\right)$ time, and successfully computes $S_B'$ such that $\EMD(S_A,S_B') \leq O(\log n) \cdot \EMD_k(S_A, S_B)$  with probability at least $5/8$.
\end{corollary}
Note that this kind of scaling strategy could be applied in the Hamming distance case too, which would change the running time of \autoref{cor:ham} to $O((dnk + nk^2) \log(nd))$.

We now prove the theorem.  The communication cost of this protocol is immediate.  There are $O(\log (D_2/D_1))$ RIBLTs, each of which has $O(k)$ cells.  Each cells takes $O(d\log (|\Delta|n))$ bits to store the value, and $O(\log n)$ bits to store the key.

The computation bound is similarly straightforward.  Each of the $n$ points is hashed $\frac{k}{8 D_1 \ln(1/p)}$ times, and each item is inserted/deleted from an RIBLT $O(\log(D_2/D_1))$ times.  It takes $O(dn k)$ time for Bob to compute all of the distances between the points in $S_B$ and those in $X_B$, and then $O(nk^2)$ time to use the Hungarian method (\cite{kuhn1955hungarian}) to find the min-cost matching between $X_B$ and $S_B$.\footnote{This assumes the distances fit into a constant number of words so they can be computed on in $O(1)$ time.  If this is not the case the $nk^2$ term in the running time increases by a factor of the number of words it takes to represent a distance.}  The time to attempt decoding of the RIBLTs is dominated by the time spent constructing them.

The proof of the approximation bound comes in several steps.  In order to bound $\EMD(S_A,S_B')$, we find a matching between the points of $S_A$ and $S_B'$, and use the cost of that matching as an upper bound.  This matching consists of three pieces.  For each bucket, where a bucket in this context is the set of points hashing to the same value at level $i^*$, we choose a maximum size matching between Alice and Bob's points within the bucket.  The remaining points are those that we wish to approximately extract from $T_{i^*}$.  The total cost of the matching is then bounded by the cost of the matching within each bucket (the \emph{in-bucket-matching}), which we call $\mu$, plus the minimum cost matching between our desired extracted points plus the distance between what we wished to extract and what we actually did ($X_A$ and $X_B$). 

More formally, we identify $Z_A \subset S_A$ and $Z_B \subset S_B$ such that $|Z_A|=|Z_B|=|X_A|=|X_B|$.  $Z_A$ and $Z_B$ are the points excluded from the in-bucket-matching. Using the definition of $\EMD$, the fact that $\EMD$ obeys the triangle inequality, and the definition of $Y_B$, we find
\begin{align*}
&\EMD(S_B',S_A)=\EMD((S_B\setminus Y_B) \cup X_A, S_A)\\
&\leq \EMD(S_B\setminus Y_B, S_A \setminus Z_A) + \EMD(X_A, Z_A)\\
&\leq \EMD(S_B\setminus Z_B, S_A \setminus Z_A) + \EMD(Y_B, Z_B) + \EMD(X_A, Z_A)\\
&\leq \EMD(S_B\setminus Z_B, S_A \setminus Z_A) + \EMD(Y_B, X_B) + \EMD(X_B,Z_B) + \EMD(X_A, Z_A) \\
&\leq \EMD(S_B\setminus Z_B, S_A \setminus Z_A) + 2\cdot \EMD(X_B,Z_B) + \EMD(X_A, Z_A).
\end{align*}
By definition, $\EMD(S_A \setminus Z_A, S_B\setminus Z_B) \leq \mu$.  We prove later in this section that with probability at least $7/8$, $\EMD(X_A,Z_A) + \EMD(X_B, Z_B) \leq O(1) \cdot \mu$.   The challenge in proving this bound is that the difference between each matched pair is an error which is added to various other cell values in the RIBLT during the peeling process.  We argue that in expectation, each error is only added to a constant number of other cells, thus the expected sum of the errors on all of the extracted points is at most a constant times the cost of the in-bucket matching.  Putting these pieces together, we get that with probability at least $7/8$, $\EMD(S_A,S_B') \leq O(1) \cdot \mu$.  What remains is to find an in-bucket-matching such that $\mu = O(\alpha^{-1} \log n) \cdot \EMD_k(S_A,S_B)$ with probability at least $7/8$.

\begin{figure}
	\centering
 		\includegraphics[width=.3\textwidth]{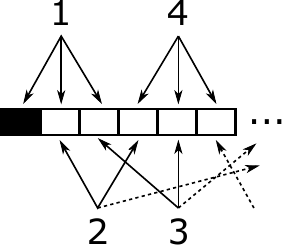}
	\caption{An example of error propagating in an (R)IBLT.  The black cell contains an error in the value, and the numbers correspond to the order in which the keys will be peeled.  The error will be added to each of the depicted cells and all four items will include it in their extracted values.}
\label{fig:error_prop_example}
\end{figure}

\begin{example} \autoref{fig:error_prop_example} shows an example of error propagating in an (R)IBLT.  The first black cell had a point from Bob and a point from Alice with different values but the same key hashed to it.  We consider these two points to be part of $Z_A$ and $Z_B$.  They canceled out all of their entries in the cell except for the difference in their values, which we call the error.  When item 1 is peeled, the point we extract will have its value offset by the error.  The peeling of 1 will add that error to the second and third cells so when items 2 and 3 are peeled their values will also be offset by the error.  The peeling of 2 and 3 will in turn propagate the error to the fourth and fifth cells so item 4 will also be extracted with the error.
\end{example}

When choosing the in-bucket-matching, whenever possible we match points that are part of the $n-k$ closest pairs in the optimal matching between $S_A$ and $S_B$ (the one defining $\EMD_k(S_A,S_B)$).  Matching these pairs costs at most $\EMD_k(S_A,S_B)$.  All that remains is to find matchings among the remaining points within the buckets that has expected cost bounded by $O(\alpha^{-1} \log n) \cdot \EMD_k(S_A,S_B)$.  Informally, we achieve this by upper bounding each points' expected matching cost by the distance from it to the furthest unmatched point in its bucket.  To do this we must first reason about $i^*$.  Going forward, we assume that the $\Theta(\log n)$-bit pairwise independent $h$ did not produce any collisions between differing MLSH vectors.  With high probability, no such collision occurs so checking equality between the hash values is equivalent to checking equality between the MLSH vectors.

We start with a simple lemma, whose proof appears in \autoref{app:emdproofs}.
\begin{restatable}{lemma}{lememdsecond}
The expected number of pairs that hash differently among the $n-k$ closest pairs in the optimal matching between $S_A$ and $S_B$ is at most $\frac{2^{i-4} k}{D_2} \EMD_k(S_A, S_B)$.
\end{restatable}

We use this lemma to bound the level at which Bob can successfully decode the RIBLT.  Once the number of pairs hashing differently is at most $k$, the RIBLT decodes successfully with high probability, so we choose $i' = \log_2\left(\frac{2D_2}{ \EMD_k}\right)$ so that the expected number of differing pairs is at most
\begin{equation*}\frac{2^{i'-4} k}{D_2} \EMD_k(S_A,S_B) = k/8.\end{equation*}
We can then use Markov's inequality to say that with probability at least $7/8$ we can decode $T_{i'}$, thus $i^* \geq i'$.  Now that we have a bound on $i^*$, we can turn to bounding $\mu$.

\begin{restatable}{lemma}{leminbucketmatch} \label{lem:inbucketmatch}
For a level $i \geq i'$, the expected value of $\mu$, the cost of the best in-bucket-matching, is $O(\alpha^{-1}\log n) \cdot \EMD_k(S_A,S_B)$.
\end{restatable}

To prove this, whenever a pair from the optimal matching (which makes up $\EMD_k(S_A,S_B)$) appear in the same bucket, we match them to each other, which contributes a total cost of at most $\EMD_k(S_A,S_B)$.  For the remaining points, we upper bound their matching cost by the maximum distance from them to every point from the other party in their bucket that is not paired with its optimal match.  In order to bound this last part, we exploit the fact that, because we are using sufficiently many MLSHs, conditioning on two points not falling into the same bucket has little impact on of whether one of those points falls into the same bucket as some specific other point, such as its optimal match.  The full proof is in \autoref{app:emdproofs}.

In the event that every matched pair in the in-bucket-matching is in fact the same point, then they would exactly cancel out and $T_{i^*}$ would be equivalent to if we only added the points from $Z_A$ and $Z_B$.  If $Z_A$ and $Z_B$ also have no duplicate keys, then the RIBLT peeling procedure would be identical to the standard IBLT peeling procedure and we would recover $Z_A$ and $Z_B$ with no error.  However, since in general the points will not be equal, when their keys cancel out, their values will leave behind some error, and when the RIBLT undergoes the peeling procedure, this error may be added to various other cells of the RIBLT.  Additionally, RIBLTs introduce error when extracting key-value pairs that have the same keys but different values, as their values are averaged (and then randomly rounded).

Let $Q_A$ and $Q_B$ be what we would recover from $T_i^*$ if the points in $Z_A$ and $Z_B$ all had different keys, and thus no averaging or rounding need to occur in their extraction.  For now we will bound the error without this averaging and rounding, then come back to it.  

We argue that with constant probability, the average number of cells a given error is added to is $O(1)$.  Since the sum of these errors is $\mu$, this implies that $\EMD(Q_A,Z_A)+\EMD(Q_B,Z_B) \leq O(1) \cdot \mu$. Note that the RIBLT does not exactly extract the value in $Z_A$ or $Z_B$ plus error, even in the case of $Q_A$ and $Q_B$ because it floors/ceilings the values back into $[0,\Delta]^d$, but this only decreases $\EMD(Q_A,Z_A)+\EMD(Q_B,Z_B)$ since $Z_A,Z_B \in [\Delta]^d$.

As discussed in \autoref{sec:iblt}, we can view the (R)IBLT peeling procedure as the process of peeling vertices of degree one from a random hypergraph.  In particular, this random hypergraph is $G_{m, cm}^q$, an $m$ vertex hypergraph with $cm$ $q$-regular hyperedges drawn uniformly at random from the $\binom{m}{q}$ possible $q$-regular hyperedges.  We model our problem as having a single random vertex initially have an ``error'', and then whenever we peel a vertex $v$, we add its error count $\EC_v$ to every adjacent vertex.  We then argue that the expected final sum of the $\EC_v$ values is $O(1)$, implying that the error only contributed to an expected constant number of extracted points as desired.  In what follows, we assume $q = O(1)$.  Since Algorithm \ref{alg:emd} required that $T_{i^*}$ decodes to at most $4k$ points and has $4q^2k$ cells, we have that $c < 1 / (q(q-1))$.  Thus, the following lemma gives us that with constant probability, $\EMD(Q_A,Z_A)+\EMD(Q_B,Z_B) \leq O(1) \cdot \mu$.

\begin{restatable}{lemma}{thmpeel} \label{thm:peel}
For $c < 1 / (q(q-1))$, after performing breadth first peeling of $G_{m,cm}^q$ we have with probability at least 7/8, $\sum_{v=1}^m {\EC_v} = O(1).$
\end{restatable}

The structure of our proof of this lemma is similar that of \cite{jiang2014parallel}, in that we relate the peeling process to an idealized branching process.  We argue that the lemma holds in this idealized branching process via careful analysis, and then argue that the branching process is sufficiently close to our peeling process that that lemma still holds there.  The full details of our proof are presented in \autoref{app:emdproofs}.

Now with our bound on $\EMD(Q_A,Z_A)+\EMD(Q_B,Z_B)$, we complete our proof of \autoref{thm:emd} with the following lemma.
\begin{restatable}{lemma}{lemaveragingerror} \label{lem:averagingerror}
With probability at least $3/4$, $$\EMD(Z_A,X_A)+\EMD(Z_B,X_B) = O(\alpha^{-1}\log n) \cdot \EMD_k(S_A,S_B).$$
\end{restatable}
The proof of this (which appears in \autoref{app:emdproofs}) follows similarly to that of \autoref{lem:inbucketmatch}.  We bound the distance between a point in $Q_A$ and the rounded average of the other points in $Q_A$ falling into its bucket, by the maximum distance between $Q_A$ and those other points from $Q_A$ in its bucket.

\section{Gap Guarantee}
\label{sec:gap}

\begin{definition}[Gap Guarantee model]
The \emph{Gap Guarantee model} of robust set reconciliation is defined for a metric space $(U,f)$ and two distance parameters $0<r_1<r_2$ as follows.  Alice and Bob have sets of points $S_A, S_B \subset U$ respectively. $|S_A| \leq n$ and $|S_B| \leq n$.  There exist subsets $C_A \subset S_A$ and $C_B \subset S_B$ such that $|C_A| \geq n - k$, $|C_B| \geq n - k$, 
$$\forall a \in C_A, \min_{b\in S_B} f(a, b) \leq r_1,$$
and
$$\forall b \in C_B, \min_{a\in S_A} f(a, b) \leq r_1.$$
The goal of the model is to minimize communication while allowing Bob to compute a set $S_B' = S_B \cup T_A$ where $T_A \subset S_A$ such that
$\forall a \in S_A, \exists b \in S_B' \text{ s.t. } f(a,b) \leq r_2.$
\end{definition}

Going forward, we refer to $a \in S_A$ and $b \in S_B$ as \emph{close} if $f(a,b) \leq r_1$ and \emph{far} if $f(a,b) \geq r_2$.  We also describe $C_A$ and $C_B$ as Alice and Bob's close points, and $T_A$ as Alice's far points.

\subsection{Our Protocol}
\label{sec:genericalg}
Our scheme hinges upon the application of a good locality sensitive hash function for our data.  Let $\mathcal{H}$ be an LSH family for our metric space with parameters $(r_1,r_2,p_1,p_2)$.  We assume that $p_2 \geq 1/2$.  The scheme operates as follows.

Each party constructs a \emph{key} for each of their elements.  A key is a vector of $h = \Theta(\log n)$ hashes.  Each of these hashes is $O(\log n)$ bits and is the evaluation of a pairwise independent hash function on a tuple of $m = \log_{p_2}(1/2)$ LSH values.
The point of the keys is that two far points have very different keys, and two close points have keys that match in all or almost all of their entries.

More concretely, we start by sampling $hm$ functions from $\mathcal{H}$.  To construct an element $x$'s key, we evaluate all of these $hm$ functions on it.  We partition these evaluations into $h$ batches of $m$.  We apply a pairwise independent hash function to each batch, and then our key is the vector of these hashes.  We then interpret each key as a set of (hash, vector index) pairs.  Alice and Bob then engage in a multisets of sets reconciliation protocol (\cite{mitzenmacher2017reconciling}) so that Alice recovers the multiset of Bob's keys.

Alice then compares their keys.  If one of her keys differs in sufficiently many of its entries  from every one of Bob's keys (the exact number depends on the parameters of the LSH and is detailed in \autoref{app:gapguarantee}), then she transmits every one of her elements that matches that key to Bob.  This protocol yields the following bounds.

\begin{restatable}{theorem}{thmgapguarantee}
\label{thm:gapguarantee}
Given a locality sensitive hash function with $\rho \leq 1 - \epsilon$ for some constant $\epsilon > 0$, there exists a protocol for the Gap Guarantee model using 4 rounds of 
\begin{equation*}O\bigg((k + \rho n) \log^2 n \left( \frac{\log n}{\log (k + \rho n) + \log \log n} + \log \log n\right)+k \log |U|\bigg)\end{equation*}
bits of communication and $O(t n \log n / \log(1 / p_2) + (k+\rho n)^2 \log^3 n)$ time, where $t$ is the time to evaluate the LSH.  This protocol succeeds with probability at least $1 - 1/n$.
\end{restatable}

The full details of the proof can be found in \autoref{app:gapguarantee}.  The first three rounds of the protocol come from a protocol for reconciling sets of sets from \cite{mitzenmacher2017reconciling}, and the final round is Alice's transmission of elements.

 Intuitively, the locality sensitive hashing here buys us two things.  First, if $\rho$ is sufficiently small then it allows us to cancel sufficiently many of the close elements without canceling the far elements so that we only need $o(n)$ communication.  Second, even when $\rho$ is large, it serves as a form of dimensionality reduction, allowing us to only transmit $O(\poly \log n)$ bits per close element, and only transmit the full $O(\log |U|)$ bits for the $k$ far elements.  While the latter could be accomplished by traditional dimensionality reduction techniques for $\ell_2$ distance \cite{johnson1984extensions}, there are some metrics, such as $\ell_1$ distance for which no sufficiently strong general dimensionality reduction scheme exists \cite{andoni2011near}.

Suppose we are working in $(\{0,1\}^n,f_H)$, the space of $n$-bit vectors under the Hamming metric.  (Note that for this example we are choosing the dimension of the space to be equal to the number of points in it).  Via the standard bit sampling LSH for Hamming distance used in \autoref{lem:hammlsh}, we see that \autoref{thm:gapguarantee} yields the optimal $O(kn)$ communication so long as $r_2/r_1=\Omega(\log^2 n \log \log n).$

\begin{corollary}
There exists a protocol for the Gap Guarantee model on $(\{0,1\}^n,f_H)$ for $r_2/r_1=\Omega(\log^2 n \log \log n)$ using 4 rounds of $O\left(k n\right)$ bits of communication and $O(n^2 \log n / r_2 + k^2 \log^3 n)$ time.  This protocol succeeds with probability at least $1 - 1/n$.
\end{corollary}

If we are working in $([\Delta]^d, \ell_1)$ with only a constant gap $r_2/r_1$, the grid LSH of \autoref{lem:gridmlsh} or the $p$-stable LSH of \autoref{lem:pstablemlsh} give us the following result.

\begin{corollary}
There exists a protocol for the Gap Guarantee model on $([\Delta]^d,\ell_1)$ for $r_2/r_1\geq 2$ using 4 rounds of $O\left(n \log^2 n \log \log n +k d \log \Delta\right)$ bits of communication and $O(d n \log n + n^2 \log^3 n)$ time.  This protocol succeeds with probability at least $1 - 1/n$.
\end{corollary}

Note that even with $r_2/r_1 = O(1)$, for large $d$ we still improve significantly over the naive solution (direct transmission) using $\Theta(n d \log \Delta)$ communication.

While this protocol works with any provided LSH, we can do slightly better in low dimensional $\ell_p$ metric spaces by using a special class of LSHs.  Specifically, we can construct an LSH with the property that $p_2 = 0$, but otherwise degrades with increased dimension.  This provides the following result.

\begin{restatable}{theorem}{thmgapguaranteelowd}
\label{thm:gapguaranteelowd}
There exists a protocol for the Gap Guarantee model on $([\Delta]^d,\ell_p)$ using 4 rounds of 
\begin{equation*}O\Bigg(\left\lceil\frac{k + \hat{\rho} n}{\log(1/\hat{\rho})}\right\rceil \log^2 n \cdot\Bigg( \frac{\log n}{\log \left\lceil\frac{k + \hat{\rho} n}{\log(1/\hat{\rho})}\right\rceil + \log \log n} + \log \log n\Bigg)+k \log |U|\Bigg)\end{equation*}
bits of communication and 
\begin{equation*}O\left(dn\left\lceil\frac{\log n}{\log(1/\hat{\rho})}\right\rceil + n \log n + \left\lceil\frac{k + \hat{\rho} n}{\log(1/\hat{\rho})}\right\rceil^2 \log^3 n\right)\end{equation*}
time, where $\hat{\rho} = r_1d/r_2$.  This protocol succeeds with probability at least $1 - 1/n$.
\end{restatable}

The full protocol and proof appears in \autoref{app:gapguarantee}.  For constant dimensional $\ell_p$ metrics with $p \in [1,2)$, this improves over \autoref{thm:gapguarantee} by roughly a factor of $\log(r_2/r_1)$ in communication.

\subsection{Lower Bound for One Round Protocols}
Our protocols for the Gap Guarantee model use four rounds of communication.  One might hope to find a protocol using only a single round, as we have for the Earth Mover's Distance model.  Our protocols could be reduced to two rounds with only a small weakening of the bounds by using a different protocol for reconciling sets of sets \cite{mitzenmacher2017reconciling}, but it is not obvious how to further reduce them to one round.  In this section we demonstrate that we cannot hope to achieve a one round protocol with competitive bounds, at least for $f_H$, the Hamming metric.

\begin{restatable}{theorem}{thmgglowerbound}
There exists no one round protocol for the Gap Guarantee on $(\{0,1\}^d,f_H)$, $d = \Omega(\log n + r_2)$, $r_1 = 1$, and $k = 1$, using $O(n)$ bits of communication that succeeds with probability at least $2/3$.
\end{restatable}

Our proof is a reduction from the index problem, and appears in \autoref{sec:gglowerbound}.  \autoref{thm:gapguarantee}, which uses more than one round, would use $O(\lceil n/r_2 \rceil \log^3 n / \log \log n + r_2)$ bits of communication in this regime, which beats this lower bound for one-round protocols when $r_2 = \omega(\log^3 n / \log \log n)$ and also $r_2 = o(n)$.  

\section{Conclusion}
Robust set reconciliation, while a very natural communication problem with several distributed system applications, has received very little study, especially compared to the standard set reconciliation problem.  We have provided new results for the EMD model utilizing IBLTs, where we analyzed error propagation during decoding;  this analysis may prove useful for other problems.  We have also considered a new variation, based on guaranteeing a small gap for all data points.  

There remains room to improve on our various results, both in terms of communication and computation, and in both the lower and upper bounds.  This work and \cite{chen2014robust} utilize IBLTs, in conjunction with various locality sensitive hashing methods, and there may be room to improve this type of combination. However, there may also remain better building block data structures available for this problem.  

\section*{Acknowledgements}
Thanks to Michael Goodrich for his insightful discussion into the applications of robust set reconciliation to machine learning settings.

\bibliographystyle{plainurl}
\bibliography{approx_set_recon}

\appendix

\section{Multiscale Locality Sensitive Hash Families}\label{app:mlsh}
\lemgridmlsh*

\begin{proof}
Our hashing scheme is to round the input points to a randomly shifted orthogonal lattice of width $w$.  The probability here of collision between $x, y \in [\Delta]^d$ is
\begin{equation*}1 - ||x - y||_1 / w \leq \Pr_{h \sim \mathcal{H}} [h(x) = h(y)] \leq (1 - ||x- y||_1 / (d w))^d,\end{equation*} assuming $||x-y||_1 \leq w$.
\begin{equation*}1 - ||x- y||_1 / w \geq e^{-2||x- y||_1 / w}\end{equation*} 
for $||x- y||_1 \leq .79 w$, and
\begin{equation*}(1 - ||x- y||_1 / (d w))^d \leq \left(e^{-||x- y||_1 / (d w)}\right)^d = e^{-||x- y||_1  / w}.\end{equation*}  Therefore, this is a an MLSH family for $([\Delta]^d, \ell_1)$ with parameters $(.79 w, e^{-2 / w}, 1/2)$.
\end{proof}

\lempstablemlsh*

\begin{proof}
We use the $p$-stable LSH scheme of \cite{datar2004locality}.  A random vector $r \in \R^d$ is chosen such that $r_1,r_2,...,r_d$ are drawn independently from a $p$-stable distribution $D$.  A distribution is $p$-stable if for any $x \in \R^d$, the distribution of $||r \cdot x||_p$ is exactly $||x||_p$ times a single draw from $D$.  In our case, we are interested in $p=2$.  The $2$-stable distribution is the Gaussian distribution, with density function $g(x) = \frac{1}{\sqrt{2\pi}}e^{-x^2/2}$.

The hashing scheme is, given an input point $x$, to output $\lfloor (r \cdot x + a) / w \rfloor$, where $w \in \R_{>0}$ and $a$ is chosen uniformly at random from $[0, w)$.  Basically, we are projecting our point into one dimension via the $p$-stable distribution, and then rounding it to a randomly shifted lattice.  The scheme results in a collision probability of
\begin{equation*}\Pr_{h\sim\mathcal{H}}[h(x)=h(y)] = 2\Phi\left(\frac{-w}{||x-y||_2}\right) - \frac{\sqrt{2} ||x-y||_2}{\sqrt{\pi} w}\left(1-e^{\frac{-w^2}{2||x-y||_2^2}}\right),\end{equation*}
where $1/2 + \Phi(x)$ is the cumulative distribution function of a Gaussian random variable.  By a Taylor expansion,
\begin{equation*}1-\frac{\sqrt{2} ||x-y||_2}{\sqrt{\pi} w} \leq \Pr_{h\sim\mathcal{H}}[h(x)=h(y)] \leq 1 - \frac{\sqrt{2} ||x-y||_2}{\sqrt{\pi} w} + e^{\frac{-w^2}{2||x-y||_2^2}} \frac{\sqrt{2} ||x-y||_2}{\sqrt{\pi} w}.\end{equation*}

\begin{equation*}1-\frac{\sqrt{2} ||x-y||_2}{\sqrt{\pi} w} \geq e^{\frac{-2\sqrt{2} ||x-y||_2}{\sqrt{\pi} w}}\end{equation*} 
for $||x- y||_2 \leq .99 w$, and
\begin{equation*}e^{\frac{-\sqrt{2} ||x-y||_2}{\sqrt{\pi} w}} + e^{\frac{-w^2}{2||x-y||_2^2}} \frac{\sqrt{2} ||x-y||_2}{\sqrt{\pi} w} \leq e^{\frac{-||x-y||_2}{2\sqrt{\pi} w}}.\end{equation*}
Thus, this is an MLSH family for $([\Delta]^d, \ell_2)$ with parameters $(.99w, e^{-2\sqrt{2/\pi}/w}, 1/(4\sqrt{2}))$.
\end{proof}

\section{Missing Proofs from the EMD Model}
\label{app:emdproofs}
Here we present the proofs omitted from the text for proving the correctness of our protocol for the Earth Mover's Distance model.

\begin{restatable}{lemma}{lememdfirst}
The probability that a pair of points at distance $x$ hash differently on level $i$ is at most 
$\frac{2^{i-4} k}{D_2} x$.
\end{restatable}
\begin{proof}
The probability that they hash differently is 
\begin{align*}
\Pr&\left[\text{at least one of $\frac{2^{i-4} k}{D_2 \ln (1/p)}$ hashes differ}\right] = 1 - \Pr\left[\text{all $\frac{2^{i-4} k}{D_2 \ln (1/p)}$ hashes match}\right] \\
&\leq 1 - p^{\frac{2^{i-4} k}{D_2 \ln (1/p)}x} = 1 - e^{-\frac{2^{i-4} k}{D_2}x} \leq \frac{2^{i-4} k}{D_2} x.
\end{align*}
\end{proof}

\lememdsecond*
\begin{proof}
Let $x_j$ be the distance between the two parties' $i$th points when they are ordered for minimum EMD and $x_j \leq x_{j+1}$.
\begin{align*}
\E&[\text{num of $n-k$ that hash differently}] = \sum_{j=1}^{n-k} \Pr\left[\text{points at distance $x_j$ hash differently}\right] \\
&\leq \sum_{j=1}^{n-k} \frac{2^{i-4} k}{D_2} x_j = \frac{2^{i-4} k}{D_2} \EMD_k(S_A,S_B).
\end{align*}
\end{proof}

\leminbucketmatch*

\begin{proof}
We will often refer to $\EMD_k(S_A,S_B)$ by simply $\EMD_k$.  Whenever a pair from the optimal matching (which makes up $\EMD_k$) appear in the same bucket, we match them to each other, which contributes a total cost of at most $\EMD_k$.  For the remaining points, we upper bound their matching cost by the maximum distance from them to every point from the other party in their bucket that is not paired with its optimal match.  We refer to Alice and Bob's $k$ points that don't appear in the optimal matching as their \emph{far} points, and the remaining points as \emph{close} points.

Our analysis is divided into three pieces: the cost of matching Bob's far points to Alice's far points, the cost of matching Bob's far points to Alice's close points (and vice versa) and the cost of matching Bob's close points to Alice's close points.  We then sum these three cases to obtain an upper bound on $\mu$.

First let's bound the cost of matching far points to far points.  Consider one of Bob's $k$ far points.  Let $y_1,\ldots,y_k$ be the distance from it to each of Alice's points, ordered such that $y_j \geq y_{j+1}$.  Let $E_j$ be the event that Alice's $j$th far point in this ordering collides with Bob's point.  Let $F_j$ be the event that Alice's far points 1 through $j-1$ do not collide with Bob's point.  The expected cost of Bob's point's matching with Alice's far points is
\begin{align*}
\E[\text{far to far matching cost}] &= \sum_{j=1}^{k} y_j \Pr[E_j \cap F_j].
\end{align*}
Let $\phi_j = \Pr[E_j \cap F_j]$.  We know that $\sum_{j=1}^k \phi_j \leq 1,$ since they are disjoint events.  We also have
\begin{equation*}\phi_j \leq \Pr[E_j] \leq p^{\alpha y_j \frac{2^{i-4} k}{D_2 \ln (1/p)}}\end{equation*}
which implies that
\begin{equation*}y_j \leq \frac{D_2\ln(1/\phi_j)}{\alpha 2^{i-4} k},\end{equation*}
and thus
\begin{align*}
\E&[\text{far to far matching cost}] \leq \sum_{j=1}^{k} \frac{D_2\ln(1/\phi_j)}{\alpha 2^{i-4} k} \phi_j \\
&= \frac{D_2}{\alpha 2^{i-4} k} \sum_{j=1}^{k} \phi_j \ln(1/\phi_j) \\
&\leq \frac{D_2}{\alpha 2^{i'-4} k} \sum_{j=1}^{k} \phi_j \ln(1/\phi_j) \\
&= \frac{8 \EMD_k}{\alpha k} \sum_{j=1}^{k} \phi_j \ln(1/\phi_j) \\
&\leq \frac{8 \EMD_k }{\alpha k} \ln k \text{  (Jensen's inequality)}.\\
\end{align*}
Summing over Bob's $k$ far points gives us $O(\alpha^{-1} \log k) \cdot \EMD_k$.

Now let's consider the cost of matching one of Bob's far points to Alice's unmatched close points.  Here $y_1,\ldots,y_{n-k}$ are the distances from Bob's point to each of Alice's close points, ordered such that $y_j \geq y_{j+1}$.  We also have $x_1,\ldots,x_{n-k}$ which are the distances from Alice's close points to their optimal matches.  Let $H_j$ be the even that Alice's $j$th close point in this ordering cannot be matched to its optimal match (they do not fall in the same bucket).  Let $E_j$ be the event that Alice's $j$th close point collides with Bob's point.  Let $F_j$ be the event that Alice's close points 1 through $j-1$ do not collide with Bob's point.  

We will use two facts there that we didn't use in the previous case.  The first is \autoref{lem:condition_on_one}, which implies that conditioning on a pair being unmatched effectively only conditions one of their MLSHs.  More specifically, 
$\Pr[E_j \; | \; H_j]$ is at most the probability that Alice's $j$th close point matches with Bob's on the first $2^{i-4}k/(D_2 \ln(1/p))-1$ hash functions.   The second fact is that, due to our bound on $p$,
\begin{equation}\label{eq:p_bound}
\frac{2^{i'-4} k}{D_2 \ln (1/p)} = \frac{k}{8\ln\left(\frac{1}{p}\right) \EMD_k} \geq 3.
\end{equation}

\begin{align*}
\E[\text{far to close matching cost}] &= \sum_{j=1}^{n-k} y_j \Pr[H_j \cap E_j \cap F_j].
\end{align*}
Let $\phi_j = \Pr[H_j \cap E_j \cap F_j]$. We know that $\sum_{j=1}^k \phi_j \leq 1,$ since they are disjoint events.  We also have
\begin{align*}
\phi_j &\leq \Pr[H_j]\Pr[E_j \; | \; H_j] \\
&\leq \frac{2^{i-4} k x_j}{D_2} p^{\alpha y_j \left(\frac{2^{i-4} k}{D_2 \ln (1/p)}-1\right)} \; \text{ (\autoref{lem:condition_on_one})} \\
&\leq \frac{2^{i-4} k x_j}{D_2} p^{\alpha y_j \frac{2^{i-5} k}{D_2 \ln (1/p)}} \; \text{ (\autoref{eq:p_bound})}.
\end{align*}
This implies that
\begin{equation*}y_j \leq \frac{D_2\ln\left(\frac{2^{i-5}k x_j}{D_2 \phi_j}\right)}{\alpha k \cdot 2^{i-2}},\end{equation*}
and thus
\begin{align*}
\E&[\text{far to close matching cost}] \leq \sum_{j=1}^{n-k} \frac{D_2\ln\left(\frac{2^{i-5}k x_j}{D_2 \phi_j}\right)}{\alpha k \cdot 2^{i-2}} \phi_j \\
&= \frac{D_2}{\alpha k \cdot 2^{i-5}} \sum_{j=1}^{n-k} \phi_j \ln\left(\frac{2^{i-4}k x_j}{D_2 \phi_j}\right) \\
&\leq \frac{D_2}{\alpha k \cdot 2^{i-5}} \ln\left(\sum_{j=1}^{n-k} \frac{2^{i-4}k x_j}{D_2}\right) \text{ (Jensen's inequality)}\\
&= \frac{D_2}{\alpha k \cdot 2^{i-5}} \ln\left(\frac{2^{i-4}k}{D_2} \EMD_k\right)\\
&\leq \frac{D_2}{\alpha k \cdot 2^{i'-5}} \ln\left(\frac{2^{i'-4}k}{D_2} \EMD_k\right)\\
&=  \frac{16 \EMD_k}{\alpha k} \ln\left(\frac{k}{8}\right).
\end{align*}
Summing over Bob's $k$ far points gives us $O(\alpha^{-1} \log k) \cdot \EMD_k$.

Finally we bound the cost of matching Bob's unmatched close points (those whose optimal match does not appear in the same bucket) to Alice's unmatched close points.  We order the points by their optimal matching, so Bob's $j$th point's optimal match is Alice's $j$th point.  Let $x_j$ be the cost of that optimal matching.  For $j,r \in [n-k]$, $m_{j,r}$ is the index of Alice's point which is the $r$th furthest from Bob's $j$th point, and $y_{j,r}$ is its distance.  Let $H_{j,r}$ be the even that Alice's $m_{j,r}$th point cannot be matched to its optimal match (they do not fall in the same bucket).  Let $E_{j,r}$ be the event that Alice's $m_{j,r}$th point collides with Bob's $j$th point.  Let $F_{j,r}$ be the event that Alice's close points $m_{j,1},\ldots,m_{j,r-1}$ do not collide with Bob's $j$th point.  Let $I_j$ be the event that Bob's $j$th point cannot be matched to its optimal match.  The expected cost of matching all of the unmatched close points is then at most
\begin{align*}
\sum_{j=1}^{n-k} &\Pr[I_j] \sum_{r=1}^{n-k} y_{j,r} \Pr[H_{j,r} \cap E_{j,r} \cap F_{j,r}  \; | \; I_j] \\
&\leq \sum_{j=1}^{n-k} \frac{2^{i-4}k x_j}{D_2} \sum_{r=1}^{n-k} y_{j,r} \Pr[H_{j,r} \cap E_{j,r} \cap F_{j,r}  \; | \; I_j].
\end{align*}
Let $\phi_{j,r} = \Pr[H_{j,r} \cap E_{j,r} \cap F_{j,r}  \; | \; I_j].$  We know that $\sum_{j=1}^k \phi_{j,r} \leq 1,$ since they are disjoint events.  We also have
\begin{align*}
\phi_{j,r} &\leq \Pr[H_{j,r} \; | \; I_j]\Pr[E_{j,r} \; | \; I_j \cap H_{j,r}] \\
&\leq 1 \cdot p^{\alpha y_{j,r} \left(\frac{2^{i-4} k}{D_2 \ln (1/p)}-2\right)} \; \text{ (\autoref{lem:condition_on_two})}\\
&\leq p^{\alpha y_{j,r} \frac{2^{i-4} k}{3D_2 \ln (1/p)}} \; \text{ (\autoref{eq:p_bound})}.
\end{align*}
This implies that
\begin{equation*}y_{j,r} \leq \frac{3D_2\ln\left(1/\phi_{j,r}\right)}{\alpha k \cdot 2^{i-4}},\end{equation*}
and thus the expected cost of matching all of the unmatched close points is upper bounded by
\begin{align*}
\sum_{j=1}^{n-k}& \frac{2^{i-4}k x_j}{D_2} \sum_{r=1}^{n-k} \frac{3D_2\ln\left(1/\phi_{j,r}\right)}{\alpha k \cdot 2^{i-4}} \phi_{j,r} \\
&= \frac{3}{\alpha} \sum_{j=1}^{n-k} x_j \sum_{r=1}^{n-k} \ln\left(1/\phi_{j,r}\right) \phi_{j,r} \\
&\leq \frac{3}{\alpha} \sum_{j=1}^{n-k} x_j \ln\left( n-k \right) \text{ (Jensen's inequality)}\\
&= \frac{3 \ln(n-k)}{\alpha} \EMD_k \\
&= O(\alpha^{-1} \log n) \cdot \EMD_k.
\end{align*}
\end{proof}

Now we turn to the main piece of our analysis, bounding the propagation of errors in the RIBLT.  Recall that we are modeling the propagation of error during the peeling process by having a single random vertex initially have an error, and then whenever we peel a vertex $v$, we add its error count $\EC_v$ to every adjacent vertex.

\thmpeel*

The following lemmas allow us to restrict our analysis to the case where the hypergraph consists of only trees and unicyclic components.  First, we have a lemma that allows us to reason about $G_{c}^q$ (a $q$-uniform hypergraph on $m$ vertices where each edge appears independently at random with probability $cm / \binom{m}{q}$) instead of $G_{m,cm}$.

\begin{lemma} \label{lem:pois}
Suppose that for all $c < 1 / (q(q-1))$, with probability at least $9/10-o(1)$, $\sum_{v=1}^m {\EC_v} = O(1)$ on $G_c^q$.  Then for all $c < 1 / (q(q-1))$, with probability at least $9/10-o(1)$, $\sum_{v=1}^m {\EC_v} = O(1)$ on $G_{m, cm}^q$.
\end{lemma}

\begin{proof}
Key here is the fact that $\sum_{v=1}^m {\EC_v}$ monotonically increases with the addition of random edges.  Let $c'$ be any constant such that $c < c' < 1 / (q(q-1))$.  With probability $1-o(1)$ $G_{c'}^q$ has more than $cm$ edges.  Therefore, if $\sum_{v=1}^m {\EC_v} = O(1)$ on $G_{c'}^q$ with probability at least $9/10-o(1)$, then $\sum_{v=1}^m {\EC_v} = O(1)$ on $G_{m,cm}^q$ with probability $9/10-o(1)$.
\end{proof}

\begin{lemma}[\cite{schmidt1985component,karonski2002phase}] \label{lem:almosttrees}
When $c < 1 / (q (q-1))$, all connected components of $G_c^q$ are either trees or unicyclic with probability $1 - O(1/n)$.
\end{lemma}

First we prove the theorem in the case that the breadth first search tree around each vertex $v$ in the graph is generated according to an idealized branching process.  Each vertex in the tree has an i.i.d. number of child edges drawn according to $\pois(cq)$, which each in turn connect to $q-1$ child vertices.  We also assume that each component is a tree, so the breadth first search tree is exactly $v$'s connected component, and then relax this assumption to include unicyclic components later.

Using this model we can show that the probability that a vertex's error propagates out to a given a radius shrinks doubly exponentially in that radius, while the number of vertices in that radius is only singly exponential, thus the expected number of vertices the error propagates to is constant.  
Later we will use the results of \cite{jiang2014parallel} to argue that this idealized branching process is sufficiently close to the true distribution and thus our results still hold.  

Recall that we are performing the peeling in a breadth first fashion.  Let $\PR_v$ be the round (starting from 1) in which vertex $v$ was peeled, and let $\survive_{v,j}$ be the event that vertex $v$ has not been peeled after $j$ rounds.  Initially only one random vertex has an error, and all the rest are $0$, so each vertex has an expected starting error of $1/m$.  Let $V_{v,r}$ be the number of vertices within radius $r$ of $v$.  Because we assume $v$'s component is a tree,
\begin{align*}
\E[\EC_v] &\leq \E[\text{number of errors in vertices within radius } \PR_v \text{ of } v] \\
&= \frac{1}{m} \E[V_{v,\PR_v}] \\
&= \frac{1}{m} \sum_{j=1}^\infty \Pr[\PR_v = j] \E[V_{v,j} \; | \; \PR_v = j] \\
&\leq \frac{1}{m} \sum_{j=1}^\infty \Pr[\survive_{v,j-1}] \E[V_{v,j} \; | \; \PR_v = j]\\
&\leq \frac{1}{m} \sum_{j=1}^\infty \Pr[\survive_{v,j-1}] \E[V_{v,j} \; | \; \survive_{v,j-1}].
\end{align*}

The idealized branching process allows us to bound $\PR[\survive_{v,j-1}]$ with the following procedure.  Consider the neighborhood of a vertex $v$ of distance $t$.  For $j=1$ up to $j = t-1$, we delete all vertices at distance $t-j$ from $v$ which have 0 child edges (those edges branching out from our process).  Then, at round $t$, we delete $v$ if after all of that it has degree at most $1$.  Let $\lambda_t$ be the probability that vertex $v$ is not deleted after this $t$ round procedure.  Although this procedure forces a certain ordering on the deletion of vertices, it is still the case that $\Pr[\survive_{v,t}] = \lambda_t$.  This is because although in the real peeling process $v$ might be deleted before $t$ rounds, if it is not deleted after $t$ rounds of this ordering, then it also will not be deleted after $t$ rounds of the real process.

Let $\rho_j$ be the probability that a vertex $u$, which is a distance $t-j$ from $v$, is not deleted after $j$ rounds.  Here $p_0 = 1$, and we find that
\[
\rho_j = \Pr[\pois(\rho_{j-1}^{q-1} cq) \geq 1].
\]
This then yields that
\[
\lambda_j = \Pr[\pois(\rho_{j-1}^{q-1} cq) \geq 2].
\]

These equations come from the fact that each node $u$ has $\pois(cq)$ child edges, each of which survives the previous round with probability $\rho_{j-1}$, independent of each other child edge.  By the splitting property of Poisson distributions \cite{mitzenmacher2005probability}, the number of surviving child edges of $u$ is thus distributed as $\pois(\rho_{j-1}^{q-1} cq)$.

For sufficiently small $c$ (in our case $c < 1/(q(q-1))$ suffices), this procedure is guaranteed to delete $v$ for sufficiently large $t$.  In other words, $\lim_{t \rightarrow \infty} \lambda_t = 0$.  \cite{jiang2014parallel} uses this fact to argue that for some constant $I > 0$ and $0 < \tau < 1$,
\[
\lambda_{I+t} \leq \tau^{2(q-1)^t},
\]
giving us a strong bound on $\Pr[\survive_{v,j-1}]$.  All that remains is to bound 
$\E[V_{v,j} \; | \; \survive_{v,j-1}].$
We use the same branching process to reason about the number of vertices at some distance $j$ from $v$.

First lets look at the expectation $\E[V_{v,j}].$ without any conditioning.  Each vertex has $\pois(cq)$ child edges, meaning
\begin{equation*}\E[V_{v,1}] = cq(q-1)+1,\end{equation*}
and an iterative application of Wald's equation \cite{wald1944cumulative} yields
\begin{equation*}\E[V_{v,t}] = \sum_{j=0}^t (cq(q-1))^j.\end{equation*}

Conditioning on $\survive_{v,j-1}$ increases the values of each of these Poisson distributions.  For $j = 2$, it means that $v$ has at least two child edges.  For $j = 3$, it means that $v$ has at least $2$ \emph{saturated} child edges.  We say a child edge is saturated if each of its $(q-1)$ vertices has at least one child edge.  For $j = 4$, the conditioning means that at least $2$ of $v$'s child edges are saturated by saturated child edges (on each of $v$'s child edge's vertices, there is at least one saturated child edge).  In general, conditioning on $j$ means that $v$ has $2$ child edges that are saturated by child edges that are saturated by child edges that are saturated by child edges, etc. $j-2$ times.  We analyze this case in \autoref{lem:conditioned_branching} and show that
\begin{equation*}\E[V_{v,j} \; | \; \survive_{v,j-1}] = O((q-1)^j).\end{equation*}

Finally putting this all together we have,
\begin{align*}
&\E[\EC_v] \leq \frac{1}{m} \sum_{j=1}^\infty \Pr[\survive_{v,j-1}] \E[V_{v,j} \; | \; \survive_{v,j-1}] \\
&\leq \frac{1}{m} \sum_{t=1}^\infty \lambda_{t-1} O((q-1)^t) \\
&= \frac{1}{m} \left(\sum_{t=1}^I \lambda_{t-1} O((q-1)^t) + \sum_{t=1}^\infty \lambda_{I+t-1} O((q-1)^{I+t})\right)\\
&\leq \frac{1}{m} \left(\sum_{t=1}^{I} O((q-1)^t) + \sum_{t=1}^{\infty} \tau^{2(q-1)^t} O((q-1)^{I+t})  \right)\\
&\leq \frac{O(1)}{m}.
\end{align*}

In the last line, $O(1)$ is using that $I, q, \phi$ and $\tau$ are all constants, together with the fact that $O((q-1)^{I+t})$ grows exponentially in $t$ while $\tau^{2(q-1)^t}$ shrinks doubly exponentially, so $\tau^{2(q-1)^t} O((q-1)^{I+t})$ converges to $0$.
Thus $\E[\EC_v] = \frac{O(1)}{m}$ for all $v$, so $\E[\sum_{v=1}^m \EC_v] = O(1)$.  

Now we address the case when a component is unicyclic.  We can bound $\E\left[\sum_{v=1}^m {\EC_v}\right]$ in this case by its value in the tree case with the addition of one additional edge, making a cycle.  Peeling this extra edge can only increase $\E\left[\sum_{v=1}^m {\EC_i}\right]$ by a factor of at most $q = O(1)$, since it only adds a single $\EC_v$ to $q-1$ other vertices, each of which ultimately contributes only $O(1)$ times its value to final sum.

Thus, if the components are all trees or unicyclic and the breadth first search tree from each vertex is generated according to the idealized branching process, then by Markov's inequality, $\sum_{v=1}^m \EC_v = O(1)$ with probability at least $9/10$.  Now we argue that the actual peeling process is sufficiently close to this idealized branching process.  To do this we need a few lemmas from \cite{jiang2014parallel}.

\begin{lemma}[Theorem 1 of \cite{jiang2014parallel}]\label{lem:parallelthm1}
Let $q \geq 3$, and let $c < c_q^*$.  With probability $1 - o(1)$, the breadth-first process of the 2-core in a random hypergraph $G_c^q$ terminates after $\log \log n / \log(q-1) + O(1)$ rounds.
\end{lemma}
Here $c_q^*$ is the threshold density below which random hypergraphs have empty 2-cores with high probability.  \cite{molloy2004pure} gives the formula for $c_q^*$ as
\begin{equation*}c_q^* = \min_{x > 0} \frac{x}{q(1-e^{-x})^{q-1}}.\end{equation*}
It is important to note that $c_q^* > 1 / (q(q-1))$, so our choice of $c$ satisfies the conditions of the theorem.

Let $E_1$ be the event that, for all vertices $v \in G_c^q$, there are at most $\log^{c_2} n$ vertices within a radius of $c_1 \log \log n$ around $v$.  There exist constants $c_1, c_2 > 0$ depending on $c$ and $q$ such that the following lemma holds.

\begin{lemma}[Lemma 3 of \cite{jiang2014parallel}]\label{lem:parallellem3}
For any event $E$, $\Pr[E] \geq \Pr[E \; | \; E_1] - 1/n$.
\end{lemma}

\begin{lemma}[Lemma 5 of \cite{jiang2014parallel}]\label{lem:parallellem5}
Let $X_1(v)$ denote the random variable describing the tree of depth $i = O(\log \log n)$ rooted at $v$ in the idealized branching process.  Let $X_2(v)$ denote the random variable describing the BFS tree of depth $i$ rooted at $v$ in $G_c^q$, conditioned on the event $E_1$ occurring.  The total variation distance between $X_1(v)$ and $X_2(v)$ is at most $\poly \log (n) / n$.
\end{lemma}

\begin{proof}[Proof of \autoref{thm:peel}]
By \autoref{lem:parallelthm1}, the peeling process completes in $O(\log \log n)$ rounds with probability $1 - o(1)$ in $G_c^q$.  This, together with our analysis of the idealized branching process, \autoref{lem:almosttrees}, and \autoref{lem:parallellem5} implies that, conditioned on $E_1$, with probability at least $9/10 - o(1) - O(1/n) - \poly \log (n) / n$, $\sum_{v=1}^m \EC_v = O(1)$ on $G_c^q$.  Putting this together with \autoref{lem:parallellem3} and \autoref{lem:pois} yields that with probability at least
\begin{equation*}9/10 - o(1) - O(1/n) - \poly \log (n) / n - 1 /n,\end{equation*} $\sum_{v=1}^m \EC_v = O(1)$ on $G_{m,cm}^q$.  For sufficiently large $n$, we have the lemma.
\end{proof}

Finally we prove the final piece to bound $\EMD(Z_A,X_A)+\EMD(Z_B,X_B)$ by reasoning about the averaging and rounding that occurs in the peeling process of RIBLTs.

\lemaveragingerror*
\begin{proof}
We will prove this for $\EMD(Z_A,X_A)$, and the argument for $\EMD(Z_B,X_B)$ is identical.

By an argument identical to that of \autoref{lem:inbucketmatch}, the expected sum over each point in $Z_A$ of the distance from that point to the furthest other point in $Z_A$ falling in the same bucket is bounded by $O(\alpha^{-1} \log n) \cdot \EMD_k(S_A,S_B)$.  By \autoref{thm:peel}, with probability at least $7/8$, the sum of these maximum distances in $Q_A$ grows by $O(\mu)$, so by \autoref{lem:inbucketmatch} the sum is still $O(\alpha^{-1} \log n) \cdot \EMD_k(S_A,S_B)$.

Now we argue that the averaging and rounding the occurs when the RIBLT extracts multiple points with the same key doesn't have too large an impact.  We will prove the case when $[\Delta]^d = \{0,1\}^d$, which intuitively can be though of as the ``hard'' case since the rounding has maximum impact here.  The general case then follows.  Let $x_1,\ldots,x_m\in \{0,1\}^d$ be our points that that hash to the same bucket, which we wish to extract.  Let $r$ be the randomized rounding of $\sum_{i=1}^m x_i / m$.  For $j \in [d]$, let $p_j = \sum_{i=1}^m x_{i,j} / m$.  $r_j = 1$ with probability $p_j$.  Let $B = \max_{i=1}^m ||x_1 - x_i||_q$.  Without loss of generality, let $x_1 = \{0\}^d$.  
\begin{align*}
\E[||x_1-r||_q] &\leq \E[||x_1-r||^q_q]^{1/q} \text{ (Jensen's inequality)} \\
&= \left(\sum_{j=1}^d p_j\right)^{1/q} \\
&\leq \left(\frac{m-1}{m}B^q\right)^{1/q} < B.
\end{align*}
The second to last inequality here used the fact that $$n \sum_{j=1}^d p_j = \sum_{i=1}^m\sum_{j=1}^d x_{i,j} \leq (m-1) B^q.$$

Thus, since the expected sum of the maximum distances is $O(\alpha^{-1} \log n) \cdot \EMD_k(S_A,S_B)$, the expected sum of distances between points and the averages of their colliding points (and thus $\EMD(Q_A,X_A)$) is $O(\alpha^{-1} \log n) \cdot \EMD_k(S_A,S_B)$.  Then by Markov's inequality and the triangle inequality, we have the lemma.
\end{proof}

\section{Conditional Probability Lemmas}
Let $X_1,\ldots,X_n$ and $Y_1,\ldots,Y_n$ be $\{0,1\}$ random variables.  We will sometimes abuse notation slightly and use $X_i$ to refer to the event that $X_i=1$. The pairs $(X_1,Y_1),\ldots,(X_n,Y_n)$ are i.i.d. and $\Pr[X_i = 1] = p$ and $\Pr[Y_i = 1] = q$.  We assume $n \geq 2$, and $0 < p,q < 1$.

\begin{lemma}
\label{lem:condition_on_one}
$\Pr\left[\cap_i X_i \; | \; \cup_i Y_i\right] \leq p^{n-1}.$
\end{lemma}

We use this in \autoref{lem:inbucketmatch} by letting $X_i$ be the event that a given pair of points ($u$ and $v$) have equal values for their $i$th MLSH function.  $Y_i$ is the event that the $v$ and $v$'s optimal matching point do not have equal values for their $i$th MLSH function.  Since each MLSH function is drawn i.i.d., these collision events are i.i.d. but correlations can exist between collision events for a given one of these functions, so this setting of $(X_i,Y_i)$ fulfills our criteria.  Thus $\Pr\left[\cap_i X_i \; | \; \cup_i Y_i\right]$ is the probability that $u$ and $v$ have all equal hash values (they land in the same bucket) given that $v$ does not have all equal hash values with its optimal match (they don't land in the same bucket).

\begin{proof}[Proof of \autoref{lem:condition_on_one}]

Since the $(X_i,Y_i)$ pairs are i.i.d., the effect of conditioning on some function of the $Y_i$s on the probability of $\cap_i X_i$ can be quantified in terms of the conditioning's effect on the distribution of $C_Y = \sum_i Y_i$.  In particular, for any function $f: \{0,1\}^n \rightarrow \{0,1\}$,
\begin{equation} \label{eq:cy}
\Pr[\cap_i X_i \; | \; f(Y_1,\ldots,Y_n)] = \sum_{j=0}^n \Pr[X_1 \; | \; Y_1]^j \Pr[X_1 \; | \; \bar{Y}_1]^{n-j} \Pr[C_Y = j \; | \; f(Y_1,\ldots,Y_n)].
\end{equation}
We will use this to relate $\Pr[\cap_i X_i \; | \; \cup_i Y_i]$ to $\Pr[\cap_i X_i]$ and $\Pr[\cap_i X_i \; | \; Y_1]$.

Let us divide our analysis in two cases, based on whether or not $X_1$ and $Y_1$ are positively correlated.  First, consider the case that $\Pr[X_1 \; | \; Y_1] \leq \Pr[X_1].$ In this case we argue that
\begin{equation*}\Pr[\cap_i X_i \; | \; \cup_i Y_i] \leq \Pr[\cap_i X_i] = p^n.\end{equation*}  
To see this, we compare $\Pr[\cap_i X_i \; | \; \cup_i Y_i]$ and $\Pr[\cap_i X_i]$ via \autoref{eq:cy}.
First note that for all $j \in [n]$, $\Pr[C_Y \geq j] \leq \Pr[C_Y \geq j \; | \; \cup_i Y_i]$.  Therefore, $f(Y_1,\ldots,Y_n)=\cup_i Y_i$ shifts the probability mass of $\Pr[C_Y = j \; | \; f(Y_1,\ldots,Y_n)]$ later in the series than $f(Y_1,\ldots,Y_n)=1$.  This shift decreases the value of the sum because $\Pr[X_1 \; | \; Y_1]^j \Pr[X_1 \; | \; \bar{Y}_1]^{n-j}$ is decreasing in $j$ (since $\Pr[X_1 \; | \; Y_1] \leq \Pr[X_1 \; | \; \bar{Y}_1]$).

Now consider the case that $\Pr[X_1 \; | \; Y_1] > \Pr[X_1].$  Here we show that \begin{equation*}\Pr[\cap_i X_i \; | \; \cup_i Y_i] \leq \Pr[\cap_i X_i \; | \; Y_1] \leq p^{n-1}.\end{equation*}

Now $\Pr[X_1 \; | \; Y_1] > \Pr[X_1 \; | \; \bar{Y}_1]$, so in \autoref{eq:cy}, $\Pr[X_1 \; | \; Y_1]^j \Pr[X_1 \; | \; \bar{Y}_1]^{n-j}$ is increasing in $j$, so it suffices to show that for all $j \in [n]$, $\Pr[C_Y \geq j \; | \; Y_1] \geq \Pr[C_Y \geq j \; | \; \cup_i Y_i]$.  Let $F_Y \in [n] \cup \{\emptyset\}$ be a random variable equal to the first index $i$ for which $Y_i=1$.  $F_Y = \emptyset$ if $Y_i = 0$ for all $i \in [n]$.  We observe that
\begin{align*}
\Pr[C_Y \geq j \; | \; \cup_i Y_i] &= \sum_{i=1}^n \Pr[C_Y \geq j \; | \; (\cup_i Y_i)\cap(F_Y = i)] \Pr[F_Y = i \; | \; \cup_i Y_i] \\
&= \sum_{i=1}^n \Pr[C_Y \geq j \; | \; F_Y = i] \Pr[F_Y = i \; | \; \cup_i Y_i].
\end{align*}
Clearly $\Pr[C_Y \geq j \; | \; F_Y = i] \geq \Pr[C_Y \geq j \; | \; F_Y = i+1]$, thus 
\begin{align*}
\Pr[C_Y \geq j \; | \; Y_1] &= \Pr[C_Y \geq j \; | \; F_Y = 1] \\
&\geq \sum_{i=1}^n \Pr[C_Y \geq j \; | \; F_Y = i] \Pr[F_Y = i \; | \; \cup_i Y_i] \\
&= \Pr[C_Y \geq j \; | \; \cup_i Y_i].
\end{align*}
\end{proof}


Now we prove an analogous lemma for the case when we are conditioning on two points missing their optimal matches.  Let $X_1,\ldots,X_n$, $Y_1,\ldots,Y_n$ and $Z_1,\ldots,Z_n$ be $\{0,1\}$ random variables.  The triples $(X_1,Y_1,Z_1),\ldots,(X_n,Y_n,Z_n)$ are i.i.d., $\Pr[X_i = 1] = p$, $\Pr[Y_i = 1] = q$ and $\Pr[Z_i = 1] = r$.  We assume $n \geq 3$, and $0 < p,q,r < 1$.

\begin{lemma}
\label{lem:condition_on_two}
$\Pr\left[\cap_i X_i \; | \; \left(\cup_i Y_i\right) \cap \left(\cup_i Z_i\right)\right] \leq p^{n-2}.$
\end{lemma}

Our application of this lemma in \autoref{lem:inbucketmatch} is very similar to that of \autoref{lem:condition_on_one}.  $X_i$ is the event that a given pair of points ($u$ and $v$) have equal values for their $i$th MLSH function.  $Y_i$ is the event that $v$ and $v$'s optimal matching point do not have equal values for their $i$th MLSH function.  $Z_i$ is the event that $u$ and $u$'s optimal match have unequal values for their $i$th MLSH function.  Since each MLSH function is drawn i.i.d., these collision events are i.i.d. but correlations can exist between collision events for a given one of these functions, so this setting of $(X_i,Y_i,Z_i)$ fulfills our criteria.  Therefore $\Pr\left[\cap_i X_i \; | \; \left(\cup_i Y_i\right) \cap \left(\cup_i Z_i\right)\right]$ is the probability that $u$ and $v$ collide across all hash values given that $u$ and $v$ do not collide with their optimal matches.

\begin{proof}[Proof of \autoref{lem:condition_on_two}]
We argue analogously to in the proof of \autoref{lem:condition_on_one}, relating \\$\Pr\left[\cap_i X_i \; | \; \left(\cup_i Y_i\right) \cap \left(\cup_i Z_i\right)\right]$ to the probability of $\cap_i X_i$ under various other conditionings of $Y_i$s and $Z_i$s.  
We partition the problem into several cases, based on the relative values of $\Pr[X_1 \; | \; Y_1 \cap Z_1]$, $\Pr[X_1 \; | \; Y_1]$, $\Pr[X_1 \; | \; Z_1]$, and $\Pr[X_1]$.

First consider the case where $Y_1$ and $Z_1$, both independently and together, reduce the probability of $X_1$.  Specifically, $\Pr[X_1] \geq \Pr[X_1 \; | \; Y_1]$,  $\Pr[X_1] \geq \Pr[X_1 \; | \; Z_1]$, and $\Pr[X_1] \geq \Pr[X_1 \; | \; Y_1 \cap Z_1]$.  In this case
\begin{equation*}\Pr\left[\cap_i X_i \; | \; \left(\cup_i Y_i\right) \cap \left(\cup_i Z_i\right)\right] \leq \Pr\left[\cap_i X_i\right] \leq p^n,\end{equation*}
because here increasing the number of $Y_i$s and $Z_i$s equal to 1 only decreases the probability of $\cap_i X_i$.

Now we look at the similar case where $\Pr[X_1] \geq \Pr[X_1 \; | \; Y_1]$ and  $\Pr[X_1] \geq \Pr[X_1 \; | \; Z_1]$, but now $\Pr[X_1] \leq \Pr[X_1 \; | \; Y_1 \cap Z_1]$.  In this case we can say
\begin{equation*}\Pr\left[\cap_i X_i \; | \; \left(\cup_i Y_i\right) \cap \left(\cup_i Z_i\right)\right] \leq \Pr\left[\cap_i X_i \; | \; \cup_i \left(Y_i \cap Z_i\right)\right] \leq \Pr\left[\cap_i X_i \; | \; Y_1 \cap Z_1 \right] \leq p^{n-1},\end{equation*}
where the first inequality is immediate from our setting, and the second inequality follows from \autoref{lem:condition_on_one}.

Next consider the case when $X_1$ is positively correlated with $Y_1$, but not with $Z_1$ (even when conditioning on $Y_1$).  That is, $\Pr[X_1 \; | \; Y_1] \geq \Pr[X_1]$, $\Pr[X_1] \geq \Pr[X_1 \; | \; Z_1]$, and $\Pr[X_1 \; | \; Y_1] \geq \Pr[X_1 \; | \; Y_1 \cap Z_1]$.  In this case conditioning on $\cup_i Z_i$ can only decrease the probability of $\cap_i X_i$, so
\begin{equation*}\Pr\left[\cap_i X_i \; | \; \left(\cup_i Y_i\right) \cap \left(\cup_i Z_i\right)\right] \leq \Pr\left[\cap_i X_i \; | \; \cup_i Y_i \right] \leq p^{n-1}.\end{equation*}

Now we examine the same case except when $\Pr[X_1 \; | \; Y_1 \cap Z_1] \geq \Pr[X_1 \; | \; Y_1]$.  That is, conditioning on $Z_1$ alone reduces the chances of $X_1$, but conditioning on $Z_1$ when $Y_1=1$ increases $X_1$'s chances.  We already know from the proof of \autoref{lem:condition_on_one} that $\Pr[\cap_i X_i \; | \; \cup_i Y_i] \leq \Pr[\cap_i X_i \; | \; Y_1]$, and it immediately follows that in this case
\begin{equation*}\Pr[\cap_i X_i \; | \; (\cup_i Y_i) \cap (\cup_i Z_i)] \leq \Pr[\cap_i X_i \; | \; Y_1 \cap Z_1] \leq p^{n-1}.\end{equation*}
The cases where $X_1$ is positively correlated with $Z_1$, but not with $Y_1$ are entirely symmetric.

In our final two cases $\Pr[X_1 \; | \; Y_1] \geq \Pr[X_1]$ and $\Pr[X_1 \; | \; Z_1] \geq \Pr[X_1]$.  First, let $\Pr[X_1 \; | \; Y_1 \cap Z_1] \geq \max(\Pr[X_1 \; | \; Y_1], \Pr[X_1 \; | \; Z_1])$.  In this case, conditioning on $Y_1$ or $Z_1$ individually increases the likelihood of $X_1$, and conditioning on both is better than either individually.  Here we argue that
\begin{align*}
\Pr[\cap_i X_i \; | \; (\cup_i Y_i) \cap (\cup_i Z_i)] &= \Pr\left[\cap_i X_i \; | \; \left(\sum_i Y_i \geq 1\right) \cap \left(\sum_i Z_i \geq 1\right)\right]\\
&\leq \Pr\left[\cap_i X_i \; | \; \sum_i \left(Y_i \cap Z_i\right) \geq 2\right]\\
&\leq \Pr[\cap_i X_i \; | \; Y_1 \cap Z_1 \cap Y_2 \cap Z_2] \\
&\leq p^{n-2}.
\end{align*}
The first inequality follows from our case parameters since conditioning on $Y_i \cap Z_i$ is stronger than conditioning on either individually and conditioning on the sum being at least two is stronger than conditioning on each individually being at least one. The second inequality follows from twice applying the argument of \autoref{lem:condition_on_one}.

Finally, let $\Pr[X_1 \; | \; Y_1 \cap Z_1] \leq \max(\Pr[X_1 \; | \; Y_1], \Pr[X_1 \; | \; Z_1])$.  Without loss of generality, let $\Pr[X_1 \; | \; Y_1] \geq \Pr[X_1 \; | \; Z_1]$.  In this case, conditioning on $Y_1$ or $Z_1$ individually increases the likelihood of $X_1$, but conditioning on both is worse than just conditioning on $Y_1$.  Here we argue that
\begin{equation*}\Pr[\cap_i X_i \; | \; (\cup_i Y_i) \cap (\cup_i Z_i)] \leq \Pr[\cap_i X_i \; | \; Y_1 \cap Y_2] \leq p^{n-2},\end{equation*}
following the logic of the previous case, except now conditioning on $Y_1$ is stronger than $Y_1 \cap Z_1$.
\end{proof}

\section{Poisson Branching Processes}
In this section we analyze $\E[V_{v,j} \; | \; \survive_{v,j-1}]$ in the model of the idealized Poisson branching process.  We build on the work of \cite{geiger1999elementary} on Galton-Watson trees, of which Poisson branching processes are a special case.  They studied the distribution of tree sizes conditioned on the tree surviving to some depth.  Our setting is similar, except that our conditioning is more complicated, relating to the fact that our process generates hypertrees.

Let $p_k=(cq)^k e^{-cq} / k!$ be the probability that a node has $k$ child edges.  Let $Z_n$ be the number of descendant vertices (recalling that each child edge connects to $q-1$ child vertices) the root has $n$ levels below it.  Let $S_n$ be the event that the root has $n$ levels of saturated child edges below it.  Here $S_1$ means that the root has at least one child edge.  $S_2$ mean the root has at least one saturated child edge.  $S_3$ means the root has at least one child edge saturated by by saturated child edges, etc.  Let $R_{n+1}$ be the index of the leftmost child edge of the root whose vertices all have $n$ levels of saturated child edges below them.  The following fact is immediate from the independence of the vertices.
\begin{fact} \label{lem:sat1}
For $n \geq 0$ and $1 \leq j \leq k < \infty$,
\begin{equation*}\Pr[R_{n+1}=j,Z_1=k \; | \; S_{n+1}] = \frac{p_k(1-\Pr[S_n]^{q-1})^{j-1}\Pr[S_n]^{q-1}}{\Pr[S_{n+1}]}.\end{equation*}
\end{fact}
Let $A_{n+1}$ be the event that the root has at least two child edges, each of whose vertices have $n$ levels of saturation.  Let $R'_{n+1}$ be the index of the second child edge (counting from the left) whose vertices have $n$ levels of saturation.
\begin{fact} \label{lem:sat2}
For $n \geq 0$ and $1 \leq j < m \leq k < \infty$,
\begin{equation*}\Pr[R_{n+1}=j,R'_{n+1}=m,Z_1=k \; | \; A_{n+1}] = \frac{p_k(1-\Pr[S_n]^{q-1})^{m-2}\Pr[S_n]^{2q-2}}{\Pr[A_{n+1}]}.\end{equation*}
\end{fact}

Using these facts, we can bound our desired quantity since
\begin{equation*}\E[\text{number of vertices within radius } j \text{ of } v \; | \; \survive_{v,j-1}] = \E\left[\sum_{i=0}^{j} Z_i \; | \; A_{j}\right].\end{equation*}

\begin{lemma}
\label{lem:conditioned_branching}
For $n \geq 0$, $c < 1 / (q (q-1))$, $q \geq 3$, and $q = O(1)$,
\begin{equation*}\E\left[\sum_{i=0}^{n+1} Z_i \; | \; A_{n+1}\right] = O((q-1)^n).\end{equation*}
\end{lemma}

\begin{proof}
The expected number of descendants up to $m$ levels below an unconditioned vertex is
\begin{equation*}\E\left[\sum_{i=0}^m Z_i\right] = \sum_{i=0}^m (cq(q-1))^i.\end{equation*}

We now use \autoref{lem:sat1} to bound this quantity conditioned on $S_{n}$.
\begin{align*}
\E&\left[\sum_{i=0}^{n+1} Z_i \; | \; S_{n+1}\right] \\
&\leq 1+\sum_{k=1}^\infty p_k \sum_{j=1}^k \frac{(1-\Pr[S_n]^{q-1})^{j-1}\Pr[S_n]^{q-1}}{\Pr[S_{n+1}]} (q-1)\left((k-1)\E\left[\sum_{i=0}^n Z_i\right]+\E\left[\sum_{i=0}^{n} Z_i \; | \; S_n\right]\right) \\
&= (q-1)\E\left[\sum_{i=0}^{n} Z_i \; | \; S_n\right]+1+(q-1)\sum_{k=1}^\infty p_k \sum_{j=1}^k \frac{(1-\Pr[S_n]^{q-1})^{j-1}\Pr[S_n]^{q-1}}{\Pr[S_{n+1}]} (k-1)\E\left[\sum_{i=0}^n Z_i\right] \\
&= (q-1)\E\left[\sum_{i=0}^{n} Z_i \; | \; S_n\right]+1+(q-1)\E\left[\sum_{i=0}^n Z_i\right]\sum_{k=1}^\infty (k-1) p_k \sum_{j=1}^k \frac{(1-\Pr[S_n]^{q-1})^{j-1}\Pr[S_n]^{q-1}}{\Pr[S_{n+1}]} \\
&= (q-1)\E\left[\sum_{i=0}^{n} Z_i \; | \; S_n\right]+1+(q-1)\E\left[\sum_{i=0}^n Z_i\right]\sum_{k=1}^\infty (k-1) p_k \sum_{j=1}^k \Pr[R_{n+1}=j \; | \; S_{n+1}, Z_1 = k] \\
&\leq (q-1)\E\left[\sum_{i=0}^{n} Z_i \; | \; S_n\right]+1+(q-1)\E\left[\sum_{i=0}^n Z_i\right]\sum_{k=1}^\infty k(k-1) p_k \\
&= (q-1)\E\left[\sum_{i=0}^{n} Z_i \; | \; S_n\right]+1+(q-1)(cq)^2\sum_{i=0}^n (cq(q-1))^i.
\end{align*}
The first inequality uses the fact that conditioning the expected number of descendants of a child edge conditioned on not all of its vertices having $n$ levels of saturation is at most the expected number of descendants of an unconditioned child edge.

We now solve this recurrence, using the fact that $\E[Z_0 \; | \; S_0] = 1$,
\begin{align*}
\E&\left[\sum_{i=0}^{n+1} Z_i \; | \; S_{n+1}\right] \leq \sum_{j=1}^{n+1} (q-1)^{n+1-j}\left(1+(q-1)(cq)^2\sum_{i=0}^j (cq(q-1))^i\right)+1\\
&= \frac{cq-1+(cq(q-1))^{n+4}+(q-1)^{n+1}\left(1-cq-c^3q^3(q-1)^2(q-2)-(q-1)^2(cq)^{n+4}\right)}{(q-2)(cq-1)(cq(q-1)-1)} \\
&= O((q-1)^n),
\end{align*}
where the final equality uses our bounds on $c$ and $q$.

Now we can use \autoref{lem:sat2} to bound the expectation conditioned on $A_n$.
\begin{align*}
&\E\left[\sum_{i=0}^{n+1} Z_i \; | \; A_{n+1}\right] \\
&\leq 1+\sum_{k=2}^\infty p_k \sum_{j=1}^{k-1} \sum_{m=j+1}^k \frac{(1-\Pr[S_n]^{q-1})^{m-2}\Pr[S_n]^{2q-2}}{\Pr[A_{n+1}]} (q-1)\\
&\qquad \cdot \left((k-2)\E\left[\sum_{i=0}^n Z_i\right]+2\E\left[\sum_{i=0}^{n} Z_i \; | \; S_n\right]\right) \\
&= 2(q-1)\E\left[\sum_{i=0}^{n} Z_i \; | \; S_n\right] + 1\\
&\qquad + (q-1)\E\left[\sum_{i=0}^n Z_i\right]\sum_{k=2}^\infty (k-2)p_k \sum_{j=1}^{k-1} \sum_{m=j+1}^k \frac{(1-\Pr[S_n]^{q-1})^{m-2}\Pr[S_n]^{2q-2}}{\Pr[A_{n+1}]}  \\
&= 2(q-1)\E\left[\sum_{i=0}^{n} Z_i \; | \; S_n\right] + 1\\
&\qquad + (q-1)\E\left[\sum_{i=0}^n Z_i\right]\sum_{k=2}^\infty (k-2)p_k \sum_{j=1}^{k-1} \sum_{m=j+1}^k \Pr[R_{n+1}=j,R'_{n+1}=m \; | \; S_{n+1}, Z_1 = k]  \\
&\leq 2(q-1)\E\left[\sum_{i=0}^{n} Z_i \; | \; S_n\right] + 1+(q-1)\E\left[\sum_{i=0}^n Z_i\right]\sum_{k=2}^\infty \frac{(k-2)(k-1)k}{2}p_k  \\
&= 2(q-1)\cdot O((q-1)^n) + 1+(q-1)(cq)^3/2\sum_{i=0}^n (cq(q-1))^i  \\
&= O((q-1)^n).
\end{align*}
Once again, the first inequality uses the fact that conditioning the expected number of descendants of a child edge conditioned on not all of its vertices having $n$ levels of saturation is at most the expected number of descendants of an unconditioned child edge.
\end{proof}

\section{Gap Guarantee Protocols}
\label{app:gapguarantee}

We use the following protocol for reconciling (multi)sets of sets.  In the multisets of sets reconciliation problem, Alice and Bob each have a parent multiset of at most $s$ child sets, each containing at most $h$ elements from a universe of size $u$.  The sum of the sizes of of all of the child sets is at most $n'$.  $\d$ is the sum over each of Alice and Bob's child sets of their minimum set difference with one of the other party's child sets.  $\q = \min(\d, s)$.  The goal of the problem is for Bob to successfully recover Alice's multiset of sets.

\begin{theorem}[Theorem 3.11 of \cite{mitzenmacher2017reconciling}]\label{thm:ssr}
Multisets of sets reconciliation can be solved in 3 rounds using \begin{equation*}O(\lceil\log_{\q}(1/\delta)\rceil \q\log s + \log (\q / \delta) \q\log h + \lceil\log_\d(1/\delta)\rceil \d \log(un'))\end{equation*} bits of communication and \begin{equation*}O(\log (\q / \delta) (n'+\q^2)+\d^2+\min(\d h, n' \sqrt{\d}, n' \log^2 h))\end{equation*} time with probability at least $1 - \delta$.
\end{theorem}

Using this, we can prove the correctness of the protocol described in \autoref{sec:genericalg}.

\thmgapguarantee*

\begin{proof}
We chose $m$ so that the probability that two far elements match on one of their hashes is at most $1/2$.  By a Chernoff bound, the probability that two far elements have keys matching in more than $h(1/2+\epsilon/6)$ entries is at most
\begin{equation*}e^{-O(h)} = 1/\poly(n)\end{equation*}
so with high probability, no pair of far keys match in more than $h(1/2+\epsilon/6)$ entries.

Now consider a close pair of elements.  The expected number matches in their keys is at least
\begin{equation*}h \cdot p_1^m = h \cdot p_1^{\log_{p_2}(1/2)} = h(1/2)^\rho \geq h (1/2)^{1-\epsilon} \geq h (1/2)^{-\log_2(1/2+\epsilon/3)} = h(1/2 + \epsilon/3),\end{equation*}
where the final inequality follows from a Taylor expansion of $-\log_2(1/2+\epsilon/3)$.

By a Chernoff bound, the probability that a close pair matches in less than $h(1/2+\epsilon/6)$ LSHs is at most
\begin{equation*}e^{-O(h)} = 1/\poly(n),\end{equation*}
so with high probability, no pair of close keys matches in less than $h(1/2+\epsilon/6)$ entries.

The total number of differences between the multisets of sets, excluding the far points, is at most
\begin{equation*}n h (1 - p_1^m) = nh(1 - (1/2)^\rho) = \Theta(nh\rho) = \Theta(\rho n \log n)\end{equation*} 
in expectation.  By a Chernoff bound, the total number is no more than $(1+c)$ times this with probability at most $e^{-O(c \rho n \log n)}$.  If $\rho > 1/n$, then this is at most $1/\poly(n)$ for $c = 1$.  If $\rho \leq 1 / n$, then this is at most $1/\poly(n)$ for $c = 1 / (\rho n)$.  Therefore, with high probability the total number of differences between the multisets of sets, including the far points, is at most 
\begin{equation*}O(\max(\rho n \log n, \log n) + k \log n) = O((k + \rho n) \log n).\end{equation*}

Now we reconcile the LSH vectors via \autoref{thm:ssr}.  Here $n' = n \log n$, $\q \leq \d = k \log n + \rho n \log n$, $s = n$, $\log u = \log n$, $h = \log n$, and $\delta = 1/n$.  We get communication
\begin{align*}
O&\left(\left\lceil\frac{\log n}{\log \d}\right\rceil \d \log n + \log(\d n) \d \log \log n\right)\\
&= O\left((k + \rho n) \log^2 n \left( \frac{\log n}{\log (k + \rho n) + \log \log n} + \log \log n\right)\right),
\end{align*}

and time 
\begin{align*}
O&(\log(\d n) (n + \d^2)+\d^2+\min(\d \log n, n \log n \sqrt{\d}, n\log n \log^2 \log n))\\
&= O(n \log n + (k+\rho n)^2 \log^3 n).
\end{align*}

Now we have already argued that, with high probability, every far key is successfully identified, and no close key is misidentified as a far key.  Note that we must have at most $k$ elements corresponding to those identified far keys.  The far keys may not be unique, but if a pair of elements corresponds to the same far key, then by our analysis both of the elements must not be close elements so by transmitting all of them, we only use $O(k \log |U|)$ communication.

The time to construct the keys is $O(t h m n) = O(t n \log n  / \log(1/p_2))$.  All that remains is to determine which keys of Alice's keys differ in at least $h(1/2+\epsilon/6)$ entries from every one of Bob's keys.  There are at most $k \log n + \rho n \log n$ keys that differ between Alice and Bob, so it takes $O((k+\rho n)^2 \log^3 n)$ to compare all of the differing keys.
\end{proof}

\subsection{Protocol for Low Dimensions} \label{sec:lowd}
While the algorithm in \autoref{sec:genericalg} works with any provided LSH, we can do slightly better in low dimensional $\ell_p$ metric spaces by using a special class of LSHs.  Specifically, we can construct an LSH with the property that $p_2 = 0$.  

The following is an LSH scheme with this kind of one-sided error that we can use for an $([\Delta]^d,\ell_p)$ metric space.  Construct a randomly shifted grid of width $r_2 / d^{1/p}$.  A point's hash value is the grid cell it falls into.   Since the maximum distance apart two points falling in the same grid cell can be is exactly $r_2$, $p_2 = 0$ as desired.  Now we bound $p_1$.

Let $x_1,\ldots,x_d$ be the absolute values of the differences between two points in each dimension.  Since we are looking at $p_1$, we want the total distance to be $r_1$ so $(\sum_{i=1}^d x_i^p)^{1/p} = r_1$.  In order for the two points to round to the different grid point, they must have at least one dimension that rounds to a different value.  By a union bound, the probability of this is at most
\begin{align*}
\sum_{i=1}^d \frac{x_i \cdot d^{1/p}}{r_2} &\leq \frac{d^{1+1/p}}{r_2} \left(\sum_{i=1}^d \frac{x_i^p}{d} \right)^{1/p} \; \text{ (Jensen's inequality)} \\
&= \frac{d}{r_2} \left(\sum_{i=1}^d x_i^p \right)^{1/p} = \frac{r_1d}{r_2}.
\end{align*}
Thus $p_1 \geq 1 - \frac{r_1d}{r_2}$.

We use the same basic protocol as before, except now this one sided LSH allows us to use $m=1$, since we don't need any replication to reduce the probability of far points colliding.  Furthermore we now only have to choose $h$ large enough that each close pair matches in at least one hash, which we achieve with $h = \Theta\left(\log n / \log\left(\frac{r_2}{r_1 d}\right)\right).$  This yields the following bound.

\thmgapguaranteelowd*

\begin{proof}
The protocol is the same as in \autoref{thm:gapguarantee}, except now $m = 1$, $h = \Theta(\lceil\log n / \log(1 / \hat{\rho})\rceil)$ and we determine one of Alice's points to be close if any of its hashes match with any of Bob's.  The analysis follows very similarly to \autoref{thm:gapguarantee}.

Consider a close pair of elements.  The probability that none of their hashes match is at most
\begin{equation*}(1 - p_1)^h = \rho^{\Theta(\lceil\log n / \log(1 / \hat{\rho})\rceil)} = 1 / \poly(n),\end{equation*}
so with high probability each pair of close elements has at least one match in their keys.  The expected number of differences in their keys is at most $h(1 - p_1) = h \hat{\rho}$.

The total number of differences between the multisets of sets, excluding the far points, is $O(\hat{\rho} n \lceil\log n / \log(1 / \hat{\rho})\rceil)$ in expectation.  By a Chernoff bound, the total number is no more than $(1+c)$ times this with probability at most $e^{-O(c \hat{\rho} n \lceil\log n / \log(1 / \hat{\rho})\rceil)}$.  If $\hat{\rho}/\log(1/\hat{\rho}) > 1/n$, then this is at most $1/\poly(n)$ for $c = 1$.  If $\hat{\rho}/\log(1/\hat{\rho}) \leq 1 / n$, then this is at most $1/\poly(n)$ for $c = \log(1/\hat{\rho}) / (\hat{\rho} n)$.  Therefore, with high probability the total number of differences between the multisets of sets, including the far points, is at most 
\begin{equation*}O(\max(\hat{\rho} n \lceil \log n / \log(1 / \hat{\rho})\rceil), \log n) + k \lceil \log n / \log(1 / \hat{\rho})\rceil) = O(\lceil(k + \hat{\rho} n) / \log(1/\hat{\rho})\rceil \log n).\end{equation*}

Now we reconcile the LSH vectors via \autoref{thm:ssr}.  Here $n' = n \lceil\log n / \log(1 / \hat{\rho})\rceil$, $\q \leq \d = \lceil(k + \hat{\rho} n) / \log(1/\hat{\rho})\rceil \log n$, $s = n$, $\log u = \log n$, $h = \lceil\log n / \log(1 / \hat{\rho})\rceil$, and $\delta = 1/n$.  We get communication
\begin{align*}
O&\left(\left\lceil\frac{\log n}{\log \d}\right\rceil \d \log n + \log(\d n) \d \log \log n\right)\\
&= O\left(\left\lceil\frac{k + \hat{\rho} n}{\log(1/\hat{\rho})}\right\rceil \log^2 n \left( \frac{\log n}{\log \left\lceil\frac{k + \hat{\rho} n}{\log(1/\hat{\rho})}\right\rceil + \log \log n} + \log \log n\right)\right),
\end{align*}

and time 
\begin{align*}
O&\left(\log(\d n) (n + \d^2)+\d^2+\min\left(\d \left\lceil\frac{\log n}{\log(1 / \hat{\rho})}\right\rceil, n \log n \sqrt{\d}, n\log n \log^2  \left\lceil\frac{\log n}{\log(1 / \hat{\rho})}\right\rceil\right)\right)\\
&=O\left(n \log n + \left\lceil\frac{k + \hat{\rho} n}{\log(1/\hat{\rho})}\right\rceil^2 \log^3 n\right).
\end{align*}

Now we have already argued that, every far key is successfully identified, and with high probability, no close key is misidentified as a far key.  Note that we must have at most $k$ elements corresponding to those identified far keys.  The far keys may not be unique, but if a pair of elements corresponds to the same far key, then by our analysis both of the elements must not be close elements so by transmitting all of them, we only use $O(k \log |U|)$ communication.

The time to construct the keys is $O(d h n) = O(d n \lceil\log n / \log(1 / \hat{\rho})\rceil)$.  All that remains is to determine which keys of Alice's keys differ in all of their entries from every one of Bob's keys.  There are at most $\lceil(k + \hat{\rho} n) / \log(1/\hat{\rho})\rceil \log n$ keys that differ between Alice and Bob, so it takes $O(\left\lceil(k + \hat{\rho} n)/\log(1/\hat{\rho})\right\rceil^2 \log^3 n)$ to compare all of the differing keys.
\end{proof}

\section{Gap Guarantee Lower Bound}
\label{sec:gglowerbound}

Here we provide the proof of our lower bound for the Gap Guarantee model:

\thmgglowerbound*

\begin{proof}
We reduce from the index problem, in which Alice has an $n$-bit string $x \in \{0,1\}^n$, Bob has an integer $i \in [n]$, and Alice wishes to send a message to Bob so that He can recover $x_i$.  The randomized communication complexity of the index problem is known to be $\Omega(n)$ \cite{kremer1999randomized}.  We reduce the index problem to the relevant form of the Gap Guarantee as follows.

In advance, the two parties agree on a set of $n+1$ $(d-1)$-bit strings $c_1,\ldots,c_{n+1}\in \{0,1\}^{d-1}$ such that for all $i \neq j \in [n+1]$, $f_H(c_i, c_j) \geq r_2$.  This is achievable for our setting of $d = \Omega(\log n + r_2)$ by a variety of error-correcting codes, such as Reed-Muller codes \cite{van2012introduction}.

Alice constructs her set of $n$ points as \begin{equation*}S_A = \{c_1 || x_1,\ldots,c_n || x_n\},\end{equation*} where $||$ is the concatenation operator.  In other words, she takes the first $n$ agreed upon codewords and appends her corresponding bit to each one.  Bob's point set is \begin{equation*}S_B = \{c_1 || 0, \ldots, c_{i-1} || 0, c_{i+1} || 0, \ldots, c_{n+1} || 0\}.\end{equation*}  He takes the set of codewords except for the $i$th one, and appends a $0$ to each one.

Now, Alice sends Bob a message so that He recovers $S_B'$ according to the Gap Guarantee definition.  $S_B'$ must contain $c_i || x_i$, and it can not have any other points within $r_2$ of that, so Bob simply finds the only new point that is at least $r_2$ from all of his original points, and reports its final bit as the solution to the index problem.
\end{proof}

\end{document}